\documentclass[aap]{imsart}

\RequirePackage{amsthm,amssymb,amsfonts,amsmath}
\RequirePackage[numbers]{natbib}
\RequirePackage[colorlinks,citecolor=blue,urlcolor=blue]{hyperref}
\RequirePackage{graphicx}

\RequirePackage[table, dvipsnames]{xcolor}
\usepackage{xfrac} 
\usepackage[most]{tcolorbox}

\usepackage{tikz}
\usetikzlibrary{calc}
\usepackage{pgfplots}
\pgfplotsset{compat=1.18}
\usetikzlibrary{trees, positioning, shapes}

\startlocaldefs

\numberwithin{equation}{section}

\theoremstyle{plain} 
\newtheorem{theorem}{Theorem}[section]
\newtheorem*{lemma*}{Lemma}
\newtheorem{lemma}[theorem]{Lemma}
\newtheorem{corollary}[theorem]{Corollary}
\newtheorem{proposition}[theorem]{Proposition}

\theoremstyle{definition} 
\newtheorem{definition}[theorem]{Definition}
\newtheorem{remark}[theorem]{Remark}

\newcommand{\oo}{\text{\o}}
\newcommand{\ooo}{\emph{\o}}
\newcommand{\vp}{\varphi}
\newcommand{\ola}[1]{\overleftarrow{#1}}

\newcommand{\R}{\mathbb{R}}
\newcommand{\N}{\mathbb{N}}

\newcommand{\eps}{\varepsilon}

\newcommand{\E}{\mathbf{E}}
\newcommand{\prob}{\mathbb{P}}

\newcommand{\I}[1]{\mathbf{1}_{\{#1\}}}

\newcommand{\hgt}{\mathrm{ht}}

\newcommand{\bbU}{\mathbb{U}}

\newcommand{\rw}{\mathrm{w}}
\newcommand{\rr}{\mathrm{r}}
\newcommand{\rk}{\mathrm{k}}
\newcommand{\dd}{\mathrm{d}}

\def\ed{\stackrel{d}{=}} 
\def\edi{\sim}
\def\aseq{\mathrel{\raisebox{-.3ex}{$\stackrel{\mathrm{a.s.}}{=}$}}} %
\def\convas{\mathrel{\raisebox{-.3ex}{$\stackrel{\mathrm{a.s.}}{\longrightarrow}$}}} %

\makeatletter
\newcommand{\myitem}[1]{%
\item[#1]\protected@edef\@currentlabel{#1}%
}
\makeatother

\endlocaldefs

\DeclareRobustCommand{\pdfd}{\texorpdfstring{$d$}{d}}

\begin{document}

\begin{frontmatter}
\title{Optimal root recovery for uniform attachment trees and \pdfd-regular growing trees}
\runtitle{Optimal root recovery for random trees}

\begin{aug}
\author[A]{\fnms{Louigi}~\snm{Addario-Berry}\ead[label=e1]{louigi.addario@mcgill.ca}}
\author[A]{\fnms{Catherine}~\snm{Fontaine}\ead[label=e2]{catherine.fontaine2@mail.mcgill.ca}}
\author[A]{\fnms{Robin}~\snm{Khanfir}\ead[label=e3]{robin.khanfir@mcgill.ca}}
\author[A]{\fnms{Louis-Roy}~\snm{Langevin}\ead[label=e4]{louis-roy.langevin@mail.mcgill.ca}}
\author[B]{\fnms{Simone}~\snm{T\^etu}\ead[label=e5]{simone\_tetu@brown.edu}}

\address[A]{Department of Mathematics and Statistics, McGill University\printead[presep={,\ }]{e1,e2,e3,e4}}

\address[B]{Division of Applied Mathematics, Brown University\printead[presep={,\ }]{e5}}
\end{aug}

\begin{abstract}
We consider root-finding algorithms for random rooted trees grown by uniform attachment. Given an unlabeled copy of the tree and a target accuracy $\eps > 0$, such an algorithm outputs a set of nodes that contains the root with probability at least $1 - \eps$. {We focus on the algorithm introduced in \cite{FindingAdam} and proved to be optimal in \cite{Inference}.} We prove that an output set of size $\exp(O(\sqrt{\log(\sfrac{1}{\eps})}))$ suffices; this bound is sharp and answers a question from \cite{FindingAdam}. We prove similar bounds for random regular trees that grow by uniform attachment, strengthening a result from~\cite{Confidence}.
\end{abstract}



\end{frontmatter}

\section{Introduction}\label{sec:intro}
A sequence $(T_n)_{n \ge 1}$ of trees is said to follow the {\em uniform attachment} or {\em UA} model if it is generated as follows: $T_1$ consists of a single node (the root, denoted by \o); for each $n \ge 1$, the tree $T_{n+1}$ is generated from $T_n$ by attaching a new node to an existing node chosen uniformly at random, independently of the previous choices. If $(T_n)_{n \ge 1}$ follows the UA model then write $T_n \edi\mathrm{UA}(n)$.
\smallskip

{\em Root finding} asks the following question: on observing $T_n$ but not the identity of the root \o, with what confidence can \o~be recovered? 
More precisely, {\em a root-finding algorithm} is a function $A$ that, given an (unlabeled, finite) input tree $T$, outputs a set $A(T)$ of nodes of $T$. The {\em size} of $A$ is the function  
$s(A):=\sup(|A(T)|,T\mbox{ a finite tree})$. The {\em error} of $A$ is $\sup(\prob(\oo \not\in A(T_n)),n \ge 1)$, where $T_n\edi \mathrm{UA}(n)$. In other words, it is the worst-case probability that the algorithm fails to return the root when the input is a uniform attachment tree of some size. 
Bubeck, Devroye and Lugosi \cite{FindingAdam} showed that for all $\eps > 0$, there exists a UA root-finding algorithm $A^{\eps}$ with error at most $\eps$ and size at most $\exp(c \log(\sfrac1\eps)/\log\log(\sfrac1\eps))$, where $c>0$ is a universal constant. They also showed that
any root finding algorithm $A$ with error at most $\eps$ must have size at least $\exp(c' \sqrt{\log\sfrac{1}{\eps}})$, for another universal constant $c'>0$, and posed the question of whether either of these bounds are tight as an open problem. 

\smallskip
The first main result of this work is to show that the above lower bound is tight, up to the value of~$c'$. 
In order to precisely state this result, we first describe the  algorithm we study. Fix a finite tree $T=(V,E)$ with $|T|:=|V|<\infty$. Define a {\em centrality measure} $\vp_T:V \to \R$ by setting
\begin{align}\label{def phi in graph}
    \varphi_T(u)=\prod_{v\in V\setminus \{u\}}|(T,u)_{v\downarrow}| 
\end{align}
for $u \in V$; here $(T,u)$ denotes the tree $T$ rooted at node $u$ and $(T,u)_{v\downarrow}$ denotes the subtree rooted at $v$ in $(T, u)$. (We say that a node $u\in V$ is \emph{more central} than a node $v\in V$ when $\varphi_T(u)\leq\varphi_T(v)$.) Next, write $n=|T|$ and order the nodes of $T$ as $(v_1,\ldots,v_n)=(v_1(T),\ldots,v_n(T))$ so that $\vp_T(v_1)\le \vp_T(v_2) \le \ldots \le \vp_T(v_n)$, breaking ties arbitrarily. Given a positive integer $k$, for an input tree $T$ with $|T|=n$, the algorithm $\mathcal{A}_k$ has output $\mathcal{A}_k(T)=\{v_1(T),\ldots,v_{k \wedge n}(T)\}$ consisting of the nodes with the $k$ smallest $\vp_T$ values (or all nodes, if the input $T$ has fewer than $k$ nodes). Clearly, $\mathcal{A}_k$ has size $k$. This algorithm was introduced by Bubeck, Devroye and Lugosi \cite{FindingAdam}, and was proved to have minimal error among algorithms of size $k$, for a fairly wide range of growing tree models, including the ones studied in this paper, in \cite[Theorem 3]{Inference}. The first main contribution of this paper is a tight analysis of the performance of this algorithm in the setting of uniform attachment trees. 
\smallskip
In the next theorem, and throughout, write $\N_1$ for the set of strictly positive integers.
\begin{theorem}\label{thm:mainua}
    There exist $C^*,c^*>0$ such that the following holds. For $\eps > 0$, let $K= K(\eps)=C^*\exp(c^*\sqrt{\log\sfrac1\eps})$. 
    Then for all $n \in \N_1$, for $T_n \edi \mathrm{UA}(n)$, it holds that $\prob(\mathrm{\ooo}\not\in \mathcal{A}_K(T_n)) \le \eps$.
\end{theorem}

The second main result of this work establishes an analogous bound for a degree-bounded variant of the UA model, previously studied in \cite{Confidence}. Given a positive integer $d \ge 3$, a {\em $d$-regular tree} is a tree in which every non-leaf node has exactly $d$ neighbours. A sequence $(T_n)_{n \ge 1}$ of trees follows the {\em $d$-regular uniform attachment} or {\em $\mathrm{UA}_d$} model if it is generated as follows. 
First, $T_1$ consists of a single root node $\oo$ and $d$ leaves. Then, for each $n \ge 1$, $T_{n+1}$ is built from $T_n$ by  adding $d-1$ new neighbours to a single leaf $V_n$ of $T_n$, with $V_n$ chosen uniformly at random from among the leaves of $T_n$, independently of all previous choices. 
Note that for all $n$, $T_n$ is a $d$-regular tree with exactly $n$ non-leaf nodes. If $(T_n)_{n \ge 1}$ follows the $\mathrm{UA}_d$ model then we write $T_n \edi \mathrm{UA}_d(n)$. 

{\color{black}
\begin{theorem}\label{thm:dary}
    There exist $C,c>0$ such that for any integer $d \ge 3$ the following holds. 
    For $\eps > 0$, let $K=K(\eps)=C\exp(c\sqrt{\log\sfrac1\eps})$. Then for all $n \in \N_1$, for $T_n \edi \mathrm{UA}_d(n)$, it holds that $\prob(\ooo\not\in \mathcal{A}_K(T_n)) \le \eps$.
\end{theorem}}
The previous theorem strengthens \cite[Corollary 2]{Confidence}, which showed that there exists $K= K(\eps)=\exp(O(\log(\sfrac{1}{\eps})/\log\log(\sfrac{1}{\eps}))) $ such that $\prob(\oo \not\in \mathcal{A}_K) \le \eps$.  Moreover, the bound in Theorem~\ref{thm:dary} is optimal up to the value of $c$; the proof of this consists in adapting a construction from \cite{FindingAdam}, in which $T_n$ grows by first building a long path stretching away from the root and then growing exclusively in the subtree at the far end of the path, to the $d$-regular setting. (The details of this construction appear in  Section~\ref{appendix_tightness}, below.)

\begin{remark}\label{remove leaves}
As defined above, the $\mathrm{UA}_d$ model slightly differs from the so-called diffusion process on the infinite $d$-regular tree studied by Khim and Loh~\cite{Confidence} and by Shah and Zaman~\cite{Rumors}. The latter is a growing sequence $(D_n)_{n\geq 1}$ of subtrees of the infinite $d$-regular tree, where each step appends a unique new vertex, chosen uniformly at random among the neighbours of $D_n$ in the infinite tree. In contrast, each step of the $\mathrm{UA}_d$ model adds $d-1$ new nodes, keeping a $d$-regular tree structure throughout. To see that the two models are equivalent, note that adding to $D_n$ all its neighbours in the infinite tree yields a $\mathrm{UA}_d$-distributed process; conversely, removing all leaves of $T_n\sim \mathrm{UA}_d(n)$ yields a random tree distributed as $D_n$. This coupling between the $\mathrm{UA}_d$ model and the diffusion process on the infinite $d$-regular tree also entails that for any fixed $k\geq 1$, any optimal root-finding algorithm of size $k$ has the same error probability for the two models. 
\end{remark}

The mathematical study of root finding was instigated in \cite{Rumors}; that work showed that the optimal size-$1$ algorithm for the $\text{UA}_3$ model has error $1/4$, and that for each $d \ge 3$ there exists $\alpha_d\in (0,1/2)$ such that the optimal size-$1$ algorithm for the $\text{UA}_3$ model has error $\alpha_d$. (Note that a size-1 algorithm must simply  output a single node as its guess for the identity of the root.) 
\smallskip

The paper \cite{FindingAdam} was the first to investigate the dependence between the error and the size; that work studied
both uniform attachment trees and {\em preferential attachment} trees, in which the probability $V_n$ connects to a node $w$ is proportional to the current degree of $w$. 
Since then, root reconstruction has been studied for a wide range of growing tree \cite{AnalysisOfCentrality,Persistence,Influence,Scaling,discovery} and graph \cite{Degree,Archaeology,looking} models. However, until quite recently, the best possible performance of a root finding algorithm (in terms of how the output set size depends on the error tolerance $\eps$) was not known for any model. The only other such tight bound we are aware of is for the preferential attachment model, for which \cite{Eve} proved that a lower bound established in \cite{FindingAdam} is in fact sharp.

\subsection{A sketch of the proof} \label{sect: sketch of proof}
We begin by sketching the proof of the following slight weakening of Theorem~\ref{thm:dary}, as it is more straightforward; the adaptations needed to prove Theorem~\ref{thm:dary} in full basically consist in establishing that certain bounds hold uniformly in $d\ge 3$. We then briefly discuss the adaptations needed to handle the uniform attachment model and thereby prove Theorem~\ref{thm:mainua}. 
{\color{black}
\begin{theorem}\label{thm:dary_weak}
    For any $d\geq 3$, there exist $C_d^*,c_d^*>0$ such that the following holds. For $\eps > 0$, let $K= K(\eps)=C_d^*\exp(c_d^*\sqrt{\log\sfrac1\eps})$. 
    Then for all $n \in \N_1$, for $T_n \edi \mathrm{UA}_d(n)$, it holds that $\prob(\ooo\not\in \mathcal{A}_K(T_n)) \le \eps$.
\end{theorem}}
The proof of Theorem~\ref{thm:dary_weak} consists of three key steps. 
\begin{enumerate}
    \myitem{(I)}\label{1_key_step} Let $D-1$ be the greatest distance from the root of $T_n$ of a node $u$ with the property that $|(T_n,\oo)_{u\downarrow}|/|T_n| \ge 1/3$. Then $D$ has exponential tails: there exists $C>0$ such that $\prob(D \ge m) \le C(\sfrac23)^m$. (See Lemma~\ref{binary_size_subtrees}, below.) 
\end{enumerate}
This first point essentially follows by a union bound. (The constant $C$ we obtain depends on $d$ in the $\mathrm{UA}_d$ model, but we believe this dependence can be straightforwardly eliminated.)  
\begin{enumerate}
    \myitem{(II)}\label{2_key_step} For a tree $T=(V,E)$ with root $\oo$, let $\Phi(T)=\vp_T(\oo)/\min_{u \in V}\vp_T(u)$. We call this the {\em competitive ratio} of $T$; it is the ratio of the centrality measure of $\oo$ to that of the most central node of $T$. Then $\Phi(T_n)$ has polynomial tails: $\prob(\Phi(T_n) \ge x) \le Ax^{-a}$ for some universal constants $A,a>0$, not depending on $d$. (See Proposition~\ref{Distribution of Phi(T)}, below.)
\end{enumerate}
To prove this, we characterize the path $\oo=u_0,\ldots,u_G=v$ from the root to the (with high probability) unique node $v$ with $\vp_{T_n}(v)$ minimal, then bound $G$, and ultimately the competitive ratio $\vp_{T_n}(\oo)/\vp_{T_n}(v)$, by a model-dependent analysis of the splitting of mass in subtrees of $T_n$. This analysis uses and builds on the ``nested P\'olya urn'' perspective introduced in \cite{FindingAdam}. The model-dependence arises here because different values of $d$ yield different urn models; but the resulting constants $a,A$ do not depend on $d$.  
\begin{enumerate}
    \myitem{(III)}\label{3_key_step} For any  integer $d \ge 3$, there exists $c=c(d) > 0$ such that for any tree $T=(V,E)$ with root $\oo$ where each node has at most $d-1$ children, if $u\in V$ is such that $|(T,\oo)_{u\downarrow}|\leq \frac{1}{3}|T|$ then $\big|\{ v \in (T,\oo)_{u\downarrow}: \varphi_{T}(\oo) \geq \varphi_{T}(v)\}\big| \leq \exp(c+c\sqrt{\log \Phi(T)})$. This deterministically bounds the number of nodes which are more competitive candidates than the root which lie in a given small subtree of $T$. (See Proposition~\ref{deterministicBound}, below.) 
\end{enumerate}
The dependence of $c$ on $d$ in \ref{3_key_step} seems unavoidable, as briefly discussed at the end of this section, and explained in greater detail later, in Section~\ref{rem:dependence_d}. However, in Section~\ref{sec:dependence_d} we explain how to remove the $d$-dependence from the constant $c_d^*$ in Theorem~\ref{thm:dary_weak} to obtain Theorem~\ref{thm:dary}.

\smallskip
The idea behind the proof of \ref{3_key_step} is that, once $|(T,\oo)_{u\downarrow}|$ is small, the function $\vp_T$ should increase quickly as one moves further into $|(T,\oo)_{u\downarrow}|$. Specifically, it is not hard to show that if $|(T,\oo)_{u\downarrow}|\le |T|/3$ then for any child $v$ of $u$, 
$\vp_T(v) \ge \vp_T(u)\cdot(|T|/|(T,\oo)_{u\downarrow}|-1) > 2|\vp_T(u)|$. 
Combining this observation with deterministic arguments similar in spirit to those in \cite{FindingAdam}, which involve some combinatorial reductions followed by the use of the Hardy-Ramanujan formula on the number of partitions of an integer, the bound in \ref{3_key_step} follows.

\smallskip

We now use \ref{1_key_step}-\ref{3_key_step} to prove Theorem~\ref{thm:dary_weak} in full, then conclude the proof sketch by briefly discussing the adaptations needed to prove Theorem~\ref{thm:mainua}. (These adaptations are rather non-trivial, and somewhat technical, and we only outline them at a high level.)

\begin{proof}[Proof of Theorem~\ref{thm:dary_weak}]
Write $B_n=\{v \in T_n:  \vp_{T_n}(\oo)\ge \vp_{T_n}(v)\}$ for the set of nodes which are at least as central as the root, and let $H_n=\{u \in T_n: |(T_n,\oo)_{u \downarrow}|\ge |T_n|/3\}$. Both $H_n\setminus B_n$ and $B_n \setminus H_n$ may be nonempty. Note that $H_n$ is a connected set of vertices which contains $\oo$, so in particular forms a subtree of $T_n$, and that by the pigeonhole principle $H_n$ contains at most three leaves, as illustrated in Figure \ref{fig:H_n}.

\begin{figure}[h]
    \centering
    \usetikzlibrary{backgrounds}

\begin{tikzpicture}
[
    dot/.style = {circle, fill=RoyalBlue!60!white, inner sep=2.2pt, draw=none},
    other dot/.style = {edge from parent/.style={draw=black, line width=1.3pt}, circle, fill=black, inner sep=2pt, draw=none},
    lobe base/.style = {fill=gray!10, draw=black, line width=0.6pt},
    pics/lobe_very_small/.style={code={ 
        \draw[lobe base] (0,0) .. controls (-0.31, 0.24) and (-0.24, 0.82) .. (0, 0.82) 
                              .. controls (0.24, 0.82) and (0.31, 0.24) .. (0,0) -- cycle;
    }},
    pics/lobe_small/.style={code={
        \draw[lobe base] (0,0) .. controls (-0.38, 0.3) and (-0.3, 1.0) .. (0, 1.0) 
                              .. controls (0.3, 1.0) and (0.38, 0.3) .. (0,0) -- cycle;
    }},
    pics/lobe_medium/.style={code={ 
        \draw[lobe base] (0,0) .. controls (-0.66, 0.52) and (-0.52, 1.73) .. (0, 1.73) 
                              .. controls (0.52, 1.73) and (0.66, 0.52) .. (0,0) -- cycle;
    }},
    pics/lobe_large/.style={code={ 
        \draw[lobe base] (0,0) .. controls (-0.93, 0.73) and (-0.73, 2.45) .. (0, 2.45) 
                              .. controls (0.73, 2.45) and (0.93, 0.73) .. (0,0) -- cycle;
    }},
    grow = up,
    level distance = 8mm,             
    level 1/.style = {level distance=15mm, sibling distance=25mm}, 
    level 2/.style = {sibling distance=22mm}, 
    level 3/.style = {sibling distance=15mm},
    level 4/.style = {sibling distance=10mm},
    level 5/.style = {sibling distance=8mm},
]

\begin{scope}[xshift=-0mm]
    \node[dot] (root1) {}
        child[other dot] { node[other dot] (n1) {} } 
        child[other dot] { 
            node[dot] (n2) {} 
            child[other dot] { 
                node[dot] (n2_child) {} 
                child[other dot] { node[other dot] (n2_l_left) {} } 
                child[other dot] { 
                    node[dot] (n2_l_right) {} 
                    child[other dot] { node[other dot] (new_l) {} } 
                    child[other dot] { 
                        node[dot] (top_blue) {} 
                        child[other dot] { node[other dot] (top_l) {} }
                        child[other dot] { node[other dot] (top_m) {} }
                        child[other dot] { node[other dot] (top_r) {} }
                    }               
                    child[other dot] { node[other dot] (new_r) {} } 
                }      
            }
        }
        child[other dot] { node[other dot] (n3) {} }
        child[other dot] { 
            node[dot] (n4) {} 
            child[other dot] { 
                node[dot] (n4_child1) {} 
                child[other dot] { node[other dot] (c1) {} } 
                child[other dot] { node[other dot] (c2) {} } 
                child[other dot] { node[other dot] (c3) {} } 
            }
            child[other dot] { node[other dot] (n4_child2) {} } 
        };
\end{scope}

\begin{scope}[on background layer]
    \path (n3) pic {lobe_medium};
    \path (n1) pic {lobe_medium};
    \path (c3) pic {lobe_medium};
    
    \path (n4_child2) pic {lobe_very_small}; 
    \path (c1) pic {lobe_small};
    \path (top_r) pic {lobe_very_small};
    \path (top_m) pic {lobe_very_small};
    \path (new_l) pic {lobe_very_small};
    \path (c2) pic {lobe_very_small}; 
    
    \foreach \n in {new_r, top_l, n2_l_left} {
        \path (\n) pic {lobe_small}; 
    }
\end{scope}

\begin{scope}[xshift=40mm, yshift=80mm]
    \node[label=right:{Nodes in $H_n$}] (l1) {};
    \node[below=5mm of l1, label=right:{Nodes in $T_n\setminus H_n$}] (l2) {};
    \node[dot] at ([xshift=-8pt]l1.west) {};
    \node[other dot] at ([xshift=-8pt]l2.west) {}; 
\end{scope}

\end{tikzpicture}
    \caption{Illustration of the set $H_n$ in a tree $T_n$. The grey lobes represent subtrees whose size is strictly less than $\frac{1}{3}n$. }
    \label{fig:H_n}
\end{figure}

For $v \in T_n \setminus H_n$, there is a unique ancestor $u$ of $v$ which is a child of an element of $H_n$ but which does not lie in $H_n$ itself. 
Since each node of $H_n$ has at most $d$ children in $T_n\setminus H_n$, we then have
\begin{align*}
|B_n| \le |H_n| + d|H_n|\max_{u \in T_n\setminus H_n} |B_n \cap (T_n,\oo)_{u \downarrow}|.
\end{align*}
Since $H_n$ has at most $3$ leaves, by the definition of $D$ in \ref{1_key_step} we have $|H_n| \le 3D$. Next, for all $u \in T_n\setminus H_n$, by \ref{3_key_step} we have $|B_n \cap (T_n,\oo)_{u \downarrow}| \le \exp(c+c\sqrt{\log\Phi(T_n)})$, so the preceding displayed bound entails that 
\[
|B_n| \le 3D\big(1+d\exp(c+c\sqrt{\log\Phi(T_n)})\big)
\le 3(d+1)D\exp(c+c\sqrt{\log\Phi(T_n)})
\, .
\]
It follows that 
\begin{align*}
 \prob\Big(|B_n| \ge 3(d+1)\log_{3/2} (\tfrac1\eps) \exp(c+c\sqrt{\tfrac{1}{a}\log\tfrac1\eps}\big)\Big) 
 & \le \prob\big(D \ge \log_{3/2}\tfrac1\eps\big) + \prob\big(\Phi(T_n) \ge \tfrac1{\eps^{1/a}} \big)\\
 & \le C \eps+A\eps\, 
\end{align*}
by the bounds in \ref{1_key_step} and \ref{2_key_step}. Since $\oo \in \mathcal{A}_K(T_n)$ whenever $K \ge |B_n|$, the result follows. 
\end{proof}

The proof of Theorem~\ref{thm:mainua} follows the same basic strategy as that of Theorem~\ref{thm:dary_weak}, but several adaptations are needed. First, to help address the fact that degrees are unbounded in the UA model, in step~\ref{1_key_step}, rather than work with {\em distances} we work with {\em weights}. Weights are defined inductively; the root has weight $\rw(\oo)=0$, and if $v$ is the $k$'th child of $u$ then we set $\rw(v)=\rw(u)+k$. With this adjustment, versions of \ref{1_key_step} and \ref{2_key_step} for the UA model follow using similar (though more technical) arguments to those for the $\mathrm{UA}_d$ model. (See respectively Lemma~\ref{Size subtree} and Proposition~\ref{Distribution of Phi(T) in UA}, below.)

\smallskip
The most challenging adaptation is that there is no fixed constant $c>0$ which makes \ref{3_key_step} true deterministically for all trees. (To see this, consider the case where $T$ is a star rooted at a leaf.) We are thus obliged to resort to probabilistic bounds; we replace the inequality in \ref{3_key_step} by something which has the flavour of the following statement: for $T_n \sim \mathrm{UA}(n)$, the ratio $\big|\{ v \in (T_n,\oo)_{u\downarrow}: \varphi_{T_n}(\oo) \geq \varphi_{T_n}(v)\}\big| /\exp(c+c\sqrt{\log \Phi(T_n)})$ has exponentially decaying upper tail, uniformly in $n$. This is not exactly what we prove, but the precise statement requires a more technical setup than is suitable for a proof overview. (See Proposition~\ref{main}.) Let us at least remark that the proof of the precise statement relies on an on-average geometric decay of the sizes of the subtrees stemming from the children of a given node; this is proved via another model-specific analysis.

\subsection{Overview of the paper}

Section~\ref{sec:setup} contains material which is used in the proofs of both main results: a formalism of the model that we use for the analysis; some basic facts about the centrality measure; an important bound for the behaviour of certain geometrically decaying flows on trees (Proposition~\ref{deter_geo_flow}); 
and some definitions, distributional identities for, and relations between different families of random variables. Section~\ref{sec:dreg} contains the proofs of the three key steps used to show Theorem~\ref{thm:dary_weak} in Section~\ref{sect: sketch of proof}, as well as a brief discussion of the $d$-dependence of the bounds of that theorem (see Section~\ref{rem:dependence_d}). Section~\ref{sec:ua} contains the proof of Theorem~\ref{thm:mainua}, followed by the adaptation of that proof which proves Theorem~\ref{thm:dary}; see Section~\ref{sec:dependence_d}. 

\section{Setup and general results}
\label{sec:setup}

This section contains formalism and results that are used in proving Theorems~\ref{thm:mainua},~\ref{thm:dary} and~\ref{thm:dary_weak}. 
The first subsection, below, introduces the {\em Ulam--Harris} formalism for rooted ordered trees, which we use for all of the analysis. 
The second subsection establishes some deterministic facts connecting centrality measure and the competitive ratio. The third subsection controls the behaviour of certain geometrically decaying functions on rooted trees, that will be crucial for our analysis. The final subsection recalls some important distributions that naturally arise in the analysis and gives some of their basic properties.

\subsection{The Ulam--Harris formalism}
\label{words_section}

Recall that $\N_1=\{1,2,3,\ldots\}$ is the set of positive integers and let $\mathbb{U}$ be the set of finite words written on the alphabet $\N_1$, namely
\[\mathbb{U}=\bigcup_{n\geq 0}\N_1^n,\quad\text{ with the convention that }\N_1^0=\{\varnothing\}.\]
For two words $u=(u_1,\ldots,u_n)$ and $v=(v_1,\ldots,v_m)$, let $u*v=(u_1,\ldots,u_n,v_1,\ldots,v_m)$ stand for the concatenation of $u$ and $v$. If $u$ is not the empty word $\varnothing$, then set $\ola{u}=(u_1,\ldots,u_{n-1})$. We further call $\ola{u}$ the \emph{parent} of $u$ and say that $u$ is a \emph{child} of $\ola{u}$. For any $j\geq 1$, we interpret $u*j:=u*(j)$ as the $j$'th child of $u$. This yields a \emph{genealogical order} denoted by $\preceq$: for $u,v\in\mathbb{U}$, we write $u\preceq v$ if and only if there exists $w\in\mathbb{U}$ such that $v=u*w$ (in other words, $u$ is a prefix of $v$). In that case, say that $u$ is an \emph{ancestor} of $v$ and that $v$ is a \emph{descendant} of $u$. Note that $u$ is always both an ancestor and a descendant of itself. We also use the notation $u \prec v$ to express that $u \preceq v$ and $u \ne v$.

\smallskip
For an integer $d \ge 2$  write $\bbU_d$ for the subset of $\bbU$ consisting of words $u=(u_1,\ldots,u_n)$ with $u_1 \in [d+1]$ and $u_i \in [d]$ for $2 \le i \le n$. The set $\bbU_{d}$ will be useful when studying the $\mathrm{UA}_{d+1}$ model.

\begin{definition}
\label{def:plane_tree}
A \emph{plane tree} is a finite set  $t\subset \mathbb{U}$ that satisfies the following properties:
\begin{enumerate}
    \item[$(a)$] it holds that $\varnothing\in t$;
    \item[$(b)$] for all $u\in t\setminus \{\varnothing\}$, it holds that $\ola{u}\in t$;
    \item[$(c)$] for all $u\in t$, there exists an integer $\rk_u(t)\geq 0$ such that for any $j\geq 1$, we have that $u*j\in t\Longleftrightarrow 1\leq j\leq \rk_u(t)$.
\end{enumerate}
\end{definition}
We view $\varnothing$ as the root of any plane tree. For a plane tree $t$ and for $u\in t$, we call the integer $\rk_u(t)$ the \emph{number of children of $u$ in $t$}. When $\rk_u(t)=0$, 
we say that $u$ is a \emph{leaf} of $t$. Finally, observe that for any $u\in t$, the set $\{v\in\mathbb{U}\, :\, u*v\in t\}$ is also a plane tree. We denote it by $\theta_u t$
and call it the \emph{(plane) subtree of $t$ stemming from $u$}. When $u\notin t$, we set $\theta_u t$ to be the empty set by convention (note that this is not a plane tree because it does not contain the empty word $\varnothing$). Subtrees  $\theta_u t$ will play essentially the same role in the remainder of the paper that the subtrees $(T,\oo)_{u \downarrow}$ played in Section~\ref{sec:intro}.
\smallskip

We say a plane tree $t$ is {\em $d$-ary} if $\rk_u(t)\in\{0,d\}$ for all $u \in t$. 
If $t$ is a $d$-ary plane tree with $m$ leaves then $|t\setminus\{\varnothing\}|=\sum_{u\in t}\rk_u(t)=(|t|-m)d$, so $(d-1)|t|=dm-1$.
\smallskip

A plane tree $t$ can be naturally viewed as a graph in the following way. Each word $u \in t$ represents a node in the graph, and the edge set consists of the pairs $\{\ola{u},u\}$, where $u \in t \setminus \{\varnothing\}$. This clearly yields a connected and acyclic graph, i.e.~a tree. By a slight abuse of notation, we identify the plane tree $t$ with its associated graph, which we also denote by $t$. Note that in the graph-theoretic sense, in a $d$-ary tree, all non-leaf nodes aside from the root have $d+1$ neighbours ($d$ children and a parent). 

\smallskip
Observe that for any $v\in t\setminus\{\varnothing\}$, removing the edge $\{\ola{v},v\}$ cuts the tree into two connected components, one of them being $\{w\in t:v\preceq w\}$. Depending on whether or not $v$ is an ancestor of $u$, we can relate either $(t,u)_{v\downarrow}$ or $(t,u)_{\ola{v}\downarrow}$ to $\theta_v t$. More precisely, we find that
\begin{itemize}
    \item if $v$ is not an ancestor of $u$, then the node set of $(t,u)_{v\downarrow}$ is $\{w\in t:v\preceq w\}$, and so $|(t,u)_{v\downarrow}|=|\theta_v t|$;
    \item if $v$ is an ancestor of $u$, then the node set of $(t,u)_{\ola{v}\downarrow}$ is the complement in $t$ of the set $\{w\in t:v\preceq w\}$, and so $|(t,u)_{\ola{v}\downarrow}|=|t|-|\theta_v t|$.
\end{itemize}
Therefore, for a plane tree $t$ and $u\in t$, the expression (\ref{def phi in graph}) of the centrality measure becomes
\begin{equation}\label{ballounes}
\varphi_t(u) = 
\prod_{v\in t,v \not\preceq u} |\theta_vt|\cdot \prod_{v \in t,\varnothing \prec v \preceq u} (|t|-|\theta_v t|)
\end{equation}
It follows from (\ref{ballounes}) that for $u\ne \varnothing$, $\vp_t(u)$ and $\vp_t(\ola{u})$ are related by the identity
\begin{equation}\label{eq:phiu-ident}
\varphi_t(u) |\theta_ut|= (|t| - |\theta_ut|) \varphi_t(\ola{u})\, ,
\end{equation}
which will be useful below. Repeatedly applying (\ref{eq:phiu-ident}) yields 
that for any vertices $u,v \in t$ with $u \preceq v$, 
\begin{equation}\label{eq:phiu-ancestor}
\frac{\varphi_t(u)}{\varphi_t(v)} = \prod_{u \prec w \preceq v} \frac{|\theta_wt|}{|t| - |\theta_wt|} .
\end{equation}
In the preceding equation, and below, we omit the constraint ``$w\in t$'' from the subscript of the product for succinctness, whenever the tree $t$ is clear from context.

\smallskip
Finally, recall the {\em competitive ratio} of $t$ defined in Section \ref{sect: sketch of proof},
\begin{equation}
\label{PHI_def}
\Phi(t)=\frac{\varphi_t(\varnothing)}{\min_{u\in t}\varphi_t(u)},
\end{equation}
a quantity which naturally arises in our approach to bounding the number of nodes of $t$ that are more central than the root.

\subsection{The centrality measure and the competitive ratio}

The following lemma gives a necessary condition on the size of subtrees stemming from vertices that are more central than the root.
\begin{lemma}\label{Condition on size of subtrees}
    For any plane tree $t$ and $u \in t$, if $|\theta_u t|/|t| < (1 + \Phi(t))^{-1}$ then $\varphi_t(u) > \varphi_t(\varnothing)$.
\end{lemma}
\begin{proof}
    Recall from \eqref{eq:phiu-ident} that $\varphi_t(u) |\theta_ut|= (|t| - |\theta_ut|) \varphi_t(\ola{u})$, so 
    \[
    \frac{\varphi_t(\varnothing)}{\varphi_t(u)}\cdot\frac{|t| - |\theta_ut|}{|\theta_ut|} 
    = \frac{\vp_t(\varnothing)}{\vp_t(\ola{u})}
    \leq \Phi(t)
    \]
    by the definition of $\Phi$. The assumption of the lemma gives $\Phi(t) < (|t| - |\theta_ut|)/|\theta_ut|$, and so $\varphi_t(\varnothing)/\varphi_t(u) < 1$.
\end{proof}
\begin{lemma}\label{Upper bound Phi}
    For a plane tree $t$, define the sequence of nodes $\varnothing = u_0 \preceq u_1 \preceq \cdots \preceq u_k$ where for each $0 \le i < k$, $u_{i+1}$ is the unique child of $u_i$ such that $|\theta_{u_{i+1}}t| \ge \frac{1}{2}|t|$ and $k \ge 0$ is the first time such a child does not exist. Then,
    \begin{equation}\label{Product of Phi}
        \Phi(t) \le \prod_{i = 1}^k \frac{1}{1 - |\theta_{u_i}t|/|t|}\,.
    \end{equation}
\end{lemma}
\begin{proof}
    We first show that $u_k$ minimizes $\varphi_t$.  Using \eqref{eq:phiu-ident}, we have that 
    $\varphi_t(\ola{u})/\varphi_t(u)< 1$ if and only if $|\theta_ut| < |t|/2$, so
    \begin{equation}
    \label{relationphi}
        \varphi_t(\ola{u}) \ge \varphi_t(u) \quad \iff \quad |\theta_ut| \ge \frac{|t|}{2}\,.
    \end{equation}
    This relation directly gives us that $\varphi_t(\varnothing) \ge \varphi_t(u_1) \ge \cdots \ge \varphi_t(u_k)$.  Next, let $w$ be any child of $u_k$ and $v$ be any descendant of $w$, and let $u_k \preceq w \preceq \cdots \preceq v$ be the unique path from $u_k$ to $v$.  Then using (\ref{relationphi}) and $|\theta_wt| < |t|/2$, we get $\varphi_t(u_k) < \varphi_t(w) < \cdots < \varphi_t(v)$.  Finally, for $0\leq i<k$, if $w$ is a child of $u_i$ other than $u_{i+1}$, then $|\theta_wt| < |t|/2$, which implies in the same way as above that any descendant $v$ of $w$ has $\varphi_t(u_{i}) < \varphi_t(v)$.  Combining these results, we readily obtain that $u_k$ minimizes $\varphi_t$. Using (\ref{eq:phiu-ancestor}), we get
    \begin{equation*}
        \Phi(t)=\frac{\varphi_t(\varnothing)}{\varphi_t(u_k)} = \prod_{i = 1}^k \frac{|\theta_{u_i}t|}{|t| - |\theta_{u_i}t|}\,.
    \end{equation*}
    We obtain the desired result after bounding $|\theta_{u_i}t|$ in the numerator above by $|t|$.
\end{proof}

\subsection{Preflows and flows on \texorpdfstring{$\bbU$}{U}}
\label{sec:flow}

We define a \emph{preflow} as a function $p:\bbU\to[0,1]$ such that $p(\varnothing)=1$ and $ \sum_{i\geq 1} p(u*i) \leq p(u)$ for all $u \in \mathbb{U}$. A \emph{flow} is a preflow for which it holds that $\sum_{i\geq 1} p(u*i) = p(u)$ for all $u \in \mathbb{U}$. We further say that a flow (or a preflow) $p$ is \emph{$d$-ary} when $p(u*j)=0$ for all $j\geq d+1$ and $u\in\mathbb{U}$.

\smallskip
Given a function $f:\bbU \to [0,\infty)$, for $x > 0$ write 
\begin{equation}
\label{def_E_and_N}
E_x(f) := \Big\{u \in \bbU: x\prod_{\varnothing \prec v \preceq u} \frac{f(v)}{2} \ge 1\Big\}\quad\text{ and}\quad N_x(f)=|E_x(f)|.
\end{equation}
Figure \ref{fig:Ex(f)} provides an example.

\begin{figure}[ht]
    \centering
    \resizebox{1.2\textwidth}{!}{%
        \begin{tikzpicture}[
  dot/.style = {circle, fill=black, inner sep=2pt, draw=none},
  circ2/.style = {circle, fill=violet!35!white, inner sep=2.2pt, draw=none},
  halfcirc/.style={
    circle, minimum size=3pt,
    inner sep = 2.8pt,
    path picture={
        \begin{scope}
            \clip (path picture bounding box.south west) rectangle
                  ($(path picture bounding box.north west)!0.5!(path picture bounding box.north east)$);
            \fill[violet!30!white] (path picture bounding box.south west) rectangle (path picture bounding box.north east);
        \end{scope}
        \begin{scope}
            \clip ($(path picture bounding box.south west)!0.5!(path picture bounding box.south east)$)
                  rectangle (path picture bounding box.north east);
            \fill[violet!60!white] (path picture bounding box.south west) rectangle (path picture bounding box.north east);
        \end{scope}
    },
  },  
  grow = up,
  level 1/.style = {sibling distance=40mm, level distance=15mm},
  level 2/.style = {sibling distance=12mm,  level distance=15mm},
  level 3/.style = {sibling distance=3mm,   level distance=15mm}
]
\begin{scope}[xshift=-15mm] 
\node[halfcirc, label={[xshift=-2mm]left:{$1$}}] (root1) {}
    child[edge from parent/.style={draw=black, line width=1.3pt}] {node[halfcirc, label={[xshift=-2mm]left:{$1/8$}}] (a) {}
      child {node[dot, label={left:{$\sfrac{1}{64}$}}] (a1) {} 
        child {node[dot] (a11) {}} 
        child {node[dot] (a12) {}}
        child {node[dot] (a13) {}}}
      child {node[dot, label={left:{$\sfrac{1}{32}$}}] (a2) {}
        child {node[dot] (a21) {}} 
        child {node[dot] (a22) {}}
        child {node[dot] (a23) {}}}
      child {node[dot, label={left:{$\sfrac{1}{16}$}}] (a3) {}
        child {node[dot] (a31) {}} 
        child {node[dot] (a32) {}}
        child {node[dot] (a33) {}}}}
    child[edge from parent/.style={draw=black, line width=1.3pt}] {node[halfcirc, label={[xshift=-2mm]left:{$1/4$}}] (b) {}
      child {node[dot, label={left:{$\sfrac{1}{32}$}}] (b1) {} 
        child {node[dot] (b11) {}} 
        child {node[dot] (b12) {}}
        child {node[dot] (b13) {}}}
      child {node[circ2, label={left:{$\sfrac{1}{16}$}}] (b2) {}
        child {node[dot] (b21) {}} 
        child {node[dot] (b22) {}}
        child {node[circ2] (b23) {}}}
      child {node[halfcirc, label={left:{$\sfrac{1}{8}$}}] (b3) {}
        child {node[circ2] (b31) {}} 
        child {node[circ2] (b32) {}}
        child {node[circ2] (b33) {}}}}
    child[edge from parent/.style={draw=black, line width=1.3pt}] {node[halfcirc, label={[xshift=-2mm]left:{$1/2$}}] (c) {}
      child {node[halfcirc, label={left:{$\sfrac{1}{16}$}}] (c1) {}
        child {node[dot] (c11) {}} 
        child {node[circ2] (c12) {}}
        child {node[circ2] (c13) {}}}
      child {node[halfcirc, label={left:{$\sfrac{1}{8}$}}] (c2) {}
        child {node[circ2] (c21) {}} 
        child {node[circ2] (c22) {}}
        child {node[circ2] (c23) {}}}
      child {node[halfcirc, label={left:{$\sfrac{1}{4}$}}] (c3) {}
        child {node[circ2] (c31) {}} 
        child {node[circ2] (c32) {}}
        child {node[circ2] (c33) {}}}};
\end{scope}
\node[right=1mm] at (a) {$\dots$}; \node[right=1mm] at (a1) {$\dots$}; 
\node[right=1mm] at (b1) {$\dots$}; \node[right=1mm] at (c1) {$\dots$}; 
\node[right=1mm] at (a11) {$\dots$}; \node[right=1mm] at (b11) {$\dots$}; 
\node[right=1mm] at (c11) {$\dots$};

\begin{scope}[shift={(1.5cm, 0.4cm)}]
    \node[halfcirc] (l1) {};
    \node[right=3pt of l1, text width=8cm] {Nodes in $E_x(f)$ for $x=2^{8}$ and $x=2^{15}$};

    \node[circ2] (l2) [below=15pt of l1] {};
    \node[right=3pt of l2, text width=5cm] {Nodes in $E_x(f)$ for $x=2^{15}$ only};
\end{scope}

\end{tikzpicture}
    }
    \caption{Illustration of $E_x(f)$ for the flow defined by $f(\varnothing) = 1$ and $f(u*j) = 2^{1-j} f(u)$, with $x = 2^8$ and $x = 2^{15}$. Vertices shown in both light and dark purple belong to both $E_{2^8}(f)$ and $E_{2^{15}}(f)$, while vertices shown only in light purple belong exclusively to $E_{2^{15}}(f)$.}
    \label{fig:Ex(f)}
\end{figure}
\smallskip
The next proposition bounds $N_x(\gamma)$ for certain geometrically decaying functions $\gamma$; it will later be used to prove upper bounds on the number of nodes which have $\varphi_t$-value smaller than or equal to the root but which lie in subtrees containing less than \sfrac13 of all nodes of $t$, for plane trees $t\subset \bbU$ and $t \subset \bbU_d$.

\begin{proposition}
\label{deter_geo_flow}
Let $\alpha\in (1,2]$ and let $\gamma=\gamma_\alpha:\bbU\to [0,1]$ be inductively defined by $\gamma(\varnothing)=1$ and $\gamma(u*j)=\alpha^{1-j}\gamma(u)$ for all $ u\in\bbU$ and $j\geq 1$. Then there is a universal constant $c>0$, that does not depend on $\alpha$, such that for any $x\geq 1$, we have $N_x(\gamma)\leq \exp(c+c\sqrt{\log_{\alpha} x})$.
\end{proposition}

\begin{proof}
Let $u=(u_1,\dots, u_h)\in\mathbb{U}$ and note that $\gamma(u)=\alpha^{-\sum_{i=1}^h (u_i-1)}$. As the ancestors $v$ of $u$ distinct from $\varnothing$ are the nodes $(u_1,\ldots,u_j)$ with $j \in [h]$, the identity $\sum_{j=1}^h \sum_{i=1}^j (u_i-1)=\sum_{i=1}^h (h+1-i)(u_i-1)$ gives us that
\begin{equation}
\label{product_deter_geo_flow}
\prod_{\varnothing \prec v\preceq u}\frac{\gamma(v)}{2}= \frac{1}{2^h}\prod_{i=1}^h \alpha^{-(h+1-i)(u_i-1)}\, .
\end{equation}
Thus, writing $u_i-1=j_i$ and taking the logarithm of the product, we obtain that
\[N_x(\gamma)=\left|\left\{(j_1,\dots, j_h)\in \N_0^h : h\in\N_0, \log_{\alpha}(2)h +\sum_{i=1}^h(h+1-i)j_i \leq\log_{\alpha} x\right\}\right|\, ,\] 
where $\N_0=\{0,1,2,\ldots\}$. Letting $n=\lfloor \log_\alpha x\rfloor$ and reversing the indexing of the $j_i$, we find the bound
\[N_x(\gamma)\leq \left|\left\{(j_1,\dots, j_h)\in \N_0^h : h\leq n, \sum_{k=1}^h kj_k \leq n\right\}\right|\, .\]
This holds because $\log_{\alpha}(2)\geq 1$, so we can require that $h\leq n$. Then, observing that for $h\leq n$, the vector $(j_1,\dots,j_h)$ has the same value of the sum $\sum_{k=1}^h kj_k$ as the vector $(j_1,\dots,j_h, 0, \dots, 0)$ with $n$ coordinates, we can partition the set according to the value of this sum and obtain that
\begin{equation} 
\label{last_step_deter_geo_flow}
N_x(\gamma)\leq \sum_{s=0}^{n} (n+1)\left|\left\{(j_1,\dots, j_n)\in \N_0^n : \sum_{k=1}^n kj_k = s\right\}\right|.
\end{equation}
To conclude the proof, we apply Erd\H{o}s' non-asymptotic version of the Hardy-Ramanujan formula on the number of partitions of an integer \cite{Erdos}, which asserts that for any $s\geq 1$,
\begin{align} \label{erdos}
    \left|\left\{(j_1,\dots, j_n)\in \N_0^n : \sum_{k=1}^n kj_k = s\right\}\right|\leq \exp\left(\pi \sqrt{2s/3}\right).
\end{align}
Inserting (\ref{erdos}) within (\ref{last_step_deter_geo_flow}) yields that $N_x(\gamma)\leq (n+1)^2\exp(\pi\sqrt{2n/3})$, which completes the proof.
\end{proof}

\subsection{Some distributional definitions and identities}\label{sub:beta}

In this section we recall the definitions of Beta and Dirichlet random variables, and some relations between them. We will make frequent use of the standard Gamma function $\Gamma(z)=\int_0^\infty t^{z-1}e^{-t}dt$. Also, we write $X \sim \mu$ to mean that $X$ has distribution $\mu$, and write $X\preceq_{\mathrm{st}} Y$ if $X$ is stochastically dominated by $Y$.
\smallskip

For $\alpha,\beta>0$, say that a random variable $B$ is \emph{$\mathrm{Beta}(\alpha,\beta)$-distributed} if it has density 
$x^{\alpha-1}(1-x)^{\beta-1}/B(\alpha,\beta)$ with respect to Lebesgue measure on $[0,1]$; here $B(\alpha,\beta)=\Gamma(\alpha)\Gamma(\beta)/\Gamma(\alpha+\beta)$ is a normalizing constant. If $B\sim \mathrm{Beta}(\alpha,\beta)$ then 
\[\E[B]=\frac{\alpha}{\alpha+\beta}\quad\text{ and }\quad\mathrm{Var}(B)=\frac{\alpha\beta}{(\alpha+\beta)^2(\alpha+\beta+1)}.\]
If $\beta=(k-1)\alpha$, then this yields that $\E[B^2]=\sfrac{(\alpha+1)}{(k(k\alpha+1))}$, which we will use below. 

For random variables $B\sim \mathrm{Beta}(\alpha,\beta)$ and $B'\sim\mathrm{Beta}(\alpha',\beta')$, then 
\begin{equation}
\label{eq:betadomination}
B' \preceq_{\mathrm{st}} B \text{ if and only if } \alpha \ge \alpha' \text{ and } \beta \le \beta'\, ;
\end{equation}
see e.g.~\cite{beta_order_old}, or~\cite[Theorem 1]{beta_order_modern}.
\smallskip

Given $\alpha_1,\ldots,\alpha_k>0$, 
a random vector $X=(X_1,\ldots,X_k)$ is \emph{$\mathrm{Dir}(\alpha_1,\ldots,\alpha_k)$-distributed} if it has density 
\[
f(x_1,\ldots,x_k)=\frac{1}{B(\alpha_1,\ldots,\alpha_k)}\prod_{i=1}^k x_i^{\alpha_i-1}
\]
with respect to $(k-1)$-dimensional Lebesgue measure on the simplex 
\[
\Delta_k  = \{(x_1,\ldots,x_k) \in [0,1]^k: x_1+\ldots+x_k=1\}\, ;
\]
here $B(\alpha_1,\ldots,\alpha_k)=(\prod_{i=1}^k\Gamma(\alpha_i))/\Gamma(\sum_{i=1}^k \alpha_i)$.
If $(X_1,\ldots,X_k)\sim \mathrm{Dir}(\alpha_1,\ldots,\alpha_k)$ 
then it holds that $X_i \sim \mathrm{Beta}(\alpha_i,\sum_{j\ne i} \alpha_j)$ for each $i\in [k]$. We write $\mathrm{Dir}_k(\alpha)$ as shorthand for the symmetric Dirichlet distribution $\mathrm{Dir}(\alpha_1,\ldots,\alpha_k)$ with $\alpha_i=\alpha$ for each $1 \le i \le k$. In this case, the density $f$ may be rewritten as 
\[
f(x_1,\ldots,x_k) = \frac{\Gamma(k\alpha)}{\Gamma(\alpha)^k} \prod_{j=1}^k x_j^{\alpha-1}\, .
\]
If $(X_1,\ldots,X_k)\sim \mathrm{Dir}_k(\alpha)$ then $X_i \sim \mathrm{Beta}(\alpha,(k-1)\alpha)$ for $i\in [k]$.

\smallskip
Finally, we say a random variable $G$ is Geometric$(p)$-distributed if $\prob(G=i)=p(1-p)^{i-1}$ for integers $i \ge 1$, and that a random variable $U$ is Uniform$[a,b]$-distributed if its law is the uniform probability distribution on $[a,b]$.

\section{Regular trees}\label{sec:dreg}

This section is devoted to gathering facts about the $\mathrm{UA}_d$ model which are required for the proof of Theorems~\ref{thm:dary} and~\ref{thm:dary_weak}. 

\subsection{Presentation and asymptotics of the model}\label{presentationAndAsymptotics}

We begin by reformulating the $d$-regular uniform attachment models within the Ulam--Harris formalism as sequences of random plane trees. It will be notationally cleaner to shift the argument by one, so we will hereafter work with $(d+1)$-regular uniform attachment trees for $d \ge 2$.
\smallskip

Fix $d \ge 2$. Let $T_1=\{\varnothing,1,\ldots,d,d+1\}$ and $T_{n+1}=T_n  \cup \{V_n*1,\ldots,V_n*d\}$ for any $n\geq 1$, where, conditionally given $(T_1,\ldots,T_n)$, $V_n$ is a uniformly random leaf of $T_n$. It is clear that the graph-isomorphism class of $T_n$ has the same law as that of the  $\mathrm{UA}_{d+1}(n)$-distributed tree from in the introduction. For the remainder of the section, {\color{black} by writing $\hat{T}_n \edi \mathrm{UA}_{d+1}(n)$ we mean that $\hat{T}_n$ is a random \emph{plane tree} with the same distribution as $T_n$} (and in particular we can write $T_n \edi \mathrm{UA}_{d+1}(n)$).

\smallskip

Note that $T_n$ has exactly $dn+2$ nodes including $(d-1)n+2$ leaves. Moreover, $T_n$ is a subset of $\bbU_d$ and for any $u\in T_n\setminus\{\varnothing\}$, $u$ has either $0$ or $d$ children in $T_n$. In other words, if $u\in T_n\setminus\{\varnothing\}$ then $\theta_u T_n$ is a $d$-ary plane tree.
Since the size of $T_n$ is deterministic, the formula (\ref{ballounes}) for $\vp_{T_n}$ suggests that studying centrality for this model will require a clear understanding of the asymptotic behaviour of the proportion of nodes of $T_n$ lying in the subtree stemming from any fixed node $u\in\bbU_d$. {\color{black} We develop such an understanding via a comparison with certain nested P\'olya urns. The same basic approach was followed in~\cite{FindingAdam} to analyze uniform attachment trees without degree restrictions, though that analysis leads to looser bounds.}
\smallskip

We claim that $\lim_{n\to\infty}|\theta_u T_n|/|T_n|$ exists almost surely for any $u\in\bbU\setminus\{\varnothing\}$. 
To see this, let $\tau=\inf\{k\in \N_1:u \in T_k\}$. If $\tau=\infty$ then $|\theta_u T_n|=0$ for all $n \in \N_1$ and the result is clear. Otherwise, for $n\geq 0$, 
colour the leaves of $T_{\tau+n}$, either in black if they are descendants of $u$, or in red if not. When a leaf $V_n$ is chosen to have $d$ new children, it ceases to be a leaf but its new children take its previous colour. Hence, conditionally given $T_\tau$, the number $M_{\tau+n}$ of leaves of $\theta_u T_{\tau+n}$ evolves as the number of black balls in a P\'olya urn starting with $1$ black ball and $(d-1)\tau+1$ red balls and where at each step, the drawn ball is returned to the urn along with $d-1$ new balls of same colour. Standard results (see e.g.~\cite[Section 6]{CMP15}) then imply that $\tfrac{1}{(d-1)n}M_{\tau+n}$ converges almost surely towards a positive random variable. Since $\theta_u T_{\tau+n}$ is a $d$-ary plane tree we have $(d-1)|\theta_u T_{\tau+n}|=dM_{\tau+n}-1$, which yields that the following limits exist almost surely for all $u\in\bbU$ and $v\in\bbU\setminus\{\varnothing\}$:

\begin{equation}
\label{def_P-U_binary}
P_u=\lim_{n\to\infty}\frac{|\theta_uT_n|}{|T_n|}\quad \text{ and }\quad D_{v}=\lim_{n\to\infty}\I{v\in T_n}\frac{|\theta_{v} T_n|}{|\theta_{\overleftarrow{v}} T_n|}\, .
\end{equation}
Note that $D_{v}$ is well-defined because if $v\in T_n$ for some $n\geq 1$, then also $\ola{v} \in T_n$, so $P_{v}$ and $P_{\ola{v}}$ are almost surely positive and $D_{v}\aseq P_{v}/P_{\overleftarrow{v}}$. This entails the following useful recursive relation
\begin{equation}
\label{formula_P2_binary}
\forall u\in\bbU\setminus\{\varnothing\},\quad    P_{u} = P_{\overleftarrow{u}}\cdot D_{u}\quad\text{ almost surely.}
\end{equation}
For example, $P_1 = D_1$ because $P_{\varnothing}=1$.  We also find that $P_{(1,1)} = D_1D_{(1,1)}$, or that $P_{(2,1,2)} = D_2D_{(2,1)}D_{(2,1,2)}$. Inductively, we obtain:
\begin{equation}
\label{formula_P1_binary}
\text{if}\quad u = (u_1,u_2,\ldots,u_h)\in\bbU,\quad\text{ then }\quad P_u = \prod_{i=1}^h D_{(u_1,\dots,u_i)}\, .
\end{equation}

 Unlike~\cite{FindingAdam}, which only uses the one-dimensional marginals, we require information on the joint distribution of the family of asymptotic proportions $(P_u)_{u\in\bbU}$ by describing that of the ratios $(D_u)_{u\in\bbU\setminus\{\varnothing\}}$. Since $T_n$ is always a subset of $\bbU_d$, we have $P_u=D_u=0$ when $u\in\bbU\setminus \bbU_d$ so we only need to give the distribution of $(D_u)_{u\in\bbU_d\setminus\{\varnothing\}}$. Thus, write $\mathbf{D}_{\varnothing}=(D_1,\ldots,D_{d},D_{d+1})$, which has $d+1$ components, and for $u \in \bbU_d\setminus\{\varnothing\}$, let $\mathbf{D}_{u}=(D_{u*1},\ldots,D_{u*d})$, which has $d$ components. Also recall from Section~\ref{sub:beta} that for $\alpha>0$ and an integer $k\geq 2$, $\mathrm{Dir}_k(\alpha)$ stands for the symmetric Dirichlet distribution of order $k$ with parameter $\alpha$.

\begin{proposition} 
\label{Independence U binary}
The random vectors $(\mathbf{D}_{u},u \in \bbU_d)$ are independent. Moreover, $\mathbf{D}_{\varnothing}\sim \mathrm{Dir}_{d+1}(\tfrac{1}{d-1})$, and for all $u \in \bbU_d\setminus\{\varnothing\}$, $\mathbf{D}_{u}\sim \mathrm{Dir}_{d}(\tfrac{1}{d-1})$.
\end{proposition}

\begin{proof}
We consider a slight variant of the $(d+1)$-regular uniform attachment where, in this new model, the root has the same number of children as all the other non-leaf vertices (and so its degree is smaller by one than those vertices). Namely, let $T_0'=\{\varnothing\}$, and for $n \ge 0$ let $T_{n+1}'=T_n'\cup\{V_n'*1,\ldots,V_n'*d\}$, where $V_n'$ is a uniformly random leaf of $T_n'$ conditionally given $(T_0',\ldots,T_n')$. Observe that $T_n'$ is a $d$-ary plane tree. Also note that $T_1'=\{\varnothing,1,\ldots,d\}$, and that $T_n'$ has exactly $1+(d-1)n$ leaves and $1+dn$ vertices. The arguments used to obtain (\ref{def_P-U_binary}) also entail that for all $u\in\bbU_d\setminus\{\varnothing\}$, the limit
\begin{equation}
\label{def_P-U_binary_variant}
D_u':=\lim_{n\to\infty}\I{u\in T_n'}\frac{|\theta_u T_n'|}{|\theta_{\overleftarrow{u}}T_n'|}
\end{equation}
exists almost surely. Note that this limit is $0$ for any node $u$ which is a descendant of $(d+1)$ since $(d+1) \not \in T_n'$ for any $n\geq 0$.  
We then set $\mathbf{D}_{u}'=(D_{u*1}',\ldots,D_{u*d}')$ for all $u \in \bbU_d$. \smallskip

Now, fix $k\in\{d,d+1\}$. We run a P\'olya urn starting with balls of colours $1,\ldots,k$, one of each colour, where after each draw, we return the drawn ball along with $d-1$ new balls of the colour drawn. For $i \in [k]$, denote by $X_n^{(i)}$ the number of balls of the $i$'th colour before the $n$'th draw, so that $X_1^{(i)}=1$ and $\sum_{i=1}^k X_n^{(i)}=k+(d-1)(n-1)$. 
Then, it is a standard result of Athreya~\cite{athreya_lost} (see also \cite[Section 6]{CMP15} for a modern approach) that
\begin{equation}
\label{cv_law_dir_Polya}
\frac{1}{(d-1)n}(X_n^{(1)},\ldots,X_n^{(k)})\xrightarrow[n\to\infty]{a.s.} (X^{(1)},\ldots,X^{(k)})\sim\mathrm{Dir}_k\big(\tfrac{1}{d-1}\big)\, .
\end{equation}

Denote by $Y_n^{(i)}$ the number of times a ball of colour $i$ is drawn before the $n$'th draw; so $X_n^{(i)}=1+(d-1)Y_n^{(i)}$ and $Y_1^{(i)}=0$ for $i \in [k]$, and $\sum_{i=1}^k Y_n^{(i)}=n-1$. Next, let $(T_n^{(1)})_{n\geq 1},\ldots,(T_n^{(k)})_{n\geq 1}$ be $k$ copies of $(T_n')_{n\geq 1}$. We assume that $(T_n^{(1)})_{n\geq 1}, \ldots, (T_n^{(k)})_{n\geq 1}$, and $(X_n^{(1)},\ldots,X_n^{(k)})_{n\geq 1}$ are independent. Then let $\hat{T}_0=\{\varnothing\}$, and for $n \ge 1$ define 
\[
\hat{T}_n = \{\varnothing\}\cup \bigcup_{i=1}^k \big\{i*w\, :\, w\in T^{(i)}_{Y_n^{(i)}}\big\},
\]
so for $n \ge 1$ and $i \in [k]$ we have $\theta_i \hat{T}_n=T^{(i)}_{Y_n^{(i)}}$; it follows that $\theta_i \hat{T}_n$ has $X_n^{(i)}$ leaves. 
It is then straightforward to check by induction that if $k=d+1$ then $(\hat{T}_n)_{n\geq 1}$ has the same law as $(T_n)_{n\geq 1}$, and that if $k=d$ then $(\hat{T}_n)_{n\geq 1}$ has the same law as $(T_n')_{n\geq 1}$. 
Writing $\theta_{i*u}\hat{T}_n=\theta_{u}\theta_i \hat{T}_n=\theta_u T_{Y_n^{(i)}}^{(i)}$ for any $u\in\bbU_d\setminus\{\varnothing\}$, then applying (\ref{cv_law_dir_Polya}), we deduce that 
\[\lim_{n\to\infty}\I{i*u\in\hat{T}_n}\frac{|\theta_{i*u}\hat{T}_n|}{|\theta_{i*\overleftarrow{u}}\hat{T}_n|}=\lim_{n\to\infty}\I{u\in T_n^{(i)}}\frac{|\theta_u T_n^{(i)}|}{|\theta_{\overleftarrow{u}} T_n^{(i)}|}\quad\text{ and }\quad \lim_{n\to\infty}\frac{|\theta_{i}\hat{T}_n|}{|\hat{T}_n|}=\lim_{n\to\infty}\frac{X_n^{(i)}}{(d-1)n}\,.\]
By independence of $(X_n^{(1)},\ldots,X_n^{(k)}), (T_n^{(1)}),\ldots, (T_n^{(k)})$, it follows from (\ref{def_P-U_binary}) and (\ref{def_P-U_binary_variant}) that:
\begin{itemize}
    \item $\mathbf{D}_{\varnothing},(\mathbf{D}_{1*u})_{u\in\bbU_d},\ldots, (\mathbf{D}_{d*u})_{u\in\bbU_d},(\mathbf{D}_{(d+1)*u})_{u\in\bbU_d}$ are independent;
    \item $\mathbf{D}_{\varnothing}',(\mathbf{D}_{1*u}')_{u\in\bbU_d},\ldots, (\mathbf{D}_{d*u}')_{u\in\bbU_d}$ are independent.
\end{itemize}
Since $(T_n^{(i)})\stackrel{d}{=}(T_n')$, it also follows that $(\mathbf{D}_{i*u})_{u\in\bbU}$ and $(\mathbf{D}_{i'*u}')_{u\in\bbU_d}$ have the same law as $(\mathbf{D}_u')_{u\in\bbU}$ for all $i \in [d+1]$ and $i'\in [d]$. Moreover, (\ref{cv_law_dir_Polya}) yields that $\mathbf{D}_\varnothing \sim \mathrm{Dir}_{d+1}(\tfrac{1}{d-1})$ and that $\mathbf{D}_\varnothing'\sim\mathrm{Dir}_{d}(\tfrac{1}{d-1})$. The claimed independence and the fact that $\mathbf{D}_u\sim\mathrm{Dir}_{d}(\tfrac{1}{d-1})$ for all $u \ne \varnothing$ then follow by induction.
\end{proof}

Before stating the next results, we define the \emph{height} of $u=(u_1,\ldots,u_n) \in \bbU$ to be $\hgt(u) = n$, so that $\hgt(\varnothing)=0$ and $\hgt(u)=1+\hgt(\ola{u})$ when $u\neq\varnothing$. Remark from (\ref{formula_P1_binary}) that $P_u$ is the product of single components of exactly $\hgt(u)$ independent Dirichlet random vectors.

\begin{corollary}\label{cor:beta_moment}
    For all $v \in \bbU_d$ with $\hgt(v) \ge 2$, $\E[D_v]=\sfrac{1}{d}$ and $\E[D_v^2] = \sfrac{1}{(2d-1)}$.
\end{corollary}
\begin{proof}
Recall from Section~\ref{sub:beta} that 
if $(X_1,\ldots,X_k)\sim \mathrm{Dir}_k(\alpha)$, then for $j \in [k]$ we have $X_j\sim \mathrm{Beta}(\alpha,(k-1)\alpha)$; so  $\mathbf{E} [X_j]=1/k$ and $\E[X_j^2]=\sfrac{(\alpha+1)}{(k(k\alpha+1))}$. Taking $\alpha=\sfrac{1}{(d-1)}$ and $k=d$, the result then follows from Proposition \ref{Independence U binary}. 
\end{proof}

To conclude this section, we use our description of the asymptotic proportions to show that they uniformly geometrically decay with the height. Combined with Lemma~\ref{Condition on size of subtrees}, this result will imply that, with high probability, the nodes of $T_n$ that are more central than the root must have bounded height; this corresponds to step \ref{1_key_step} from the proof of Theorem~\ref{thm:dary_weak} in Section~\ref{sect: sketch of proof}. 
The fact that these better candidates belong with high probability to a deterministic finite subset of $\bbU_d$ will also simplify the asymptotic analysis of the $(d+1)$-regular uniform attachment model, since it allows using only finitely many of the almost sure convergences $|\theta_u T_n|/n\longrightarrow P_u$. 
\begin{lemma}
\label{binary_size_subtrees}
Fix $d \ge 2$, and for $n\geq 1$ let $T_n \sim \mathrm{UA}_{d+1}(n)$. Then for all integers $m\geq 1$ and for all $\varepsilon\in (0,1)$, it holds that
\[
\limsup_{n\to\infty}\prob\big(\exists u\in\mathbb{U}_d\, :\, \hgt(u)\geq m, |\theta_u T_n|\geq \varepsilon n\big)\leq \tfrac{3}{2}(d+1)\varepsilon^{-2}\big(\tfrac{2}{3}\big)^{m}. 
\]
\end{lemma}
\begin{proof} Let $L(m,\varepsilon)$ be the $\limsup$ in the statement of the lemma. For $u \in \mathbb{U}_d$, if $\hgt(u) \geq m$ then $u$ has a unique ancestor $v\in\mathbb{U}_d$ at  height $\hgt(v)=m$, and necessarily $|\theta_vT_n| \ge |\theta_u T_n|$. Therefore,
\begin{equation*}
    L(m,\varepsilon) \le \limsup_{n\to\infty}\prob\left(\exists v\in\mathbb{U}_d\, :\, \hgt(v) = m, |\theta_vT_n| \ge \varepsilon n\right)\,.
\end{equation*}
Since there are only finitely many nodes $v$ with $\hgt(v) = m$ in $\mathbb{U}_d$, we obtain that
\begin{equation*}
    L(m,\varepsilon) \le \sum_{\substack{v\in\mathbb{U}_d\\ \hgt(v) = m}} \limsup_{n\to\infty}\prob\left(\frac{1}{n}|\theta_vT_n|\geq \varepsilon \right)\,.
\end{equation*}
By \eqref{formula_P1_binary}, if $v=(v_1,\ldots,v_m) \in \bbU_d$, then almost surely
\[
\frac{|\theta_v T_n|}{n} \longrightarrow P_v= \prod_{i=1}^{m} D_{(v_1,\ldots,v_i)}\, .
\]
By Proposition~\ref{Independence U binary}, the random variables in the product on the right are independent and all the variables but $D_{(v_1)}$ (which is bounded by $1$) are $\mathrm{Beta}(\sfrac{1}{(d-1)},1)$-distributed.
Since there are exactly $(d+1)d^{m-1}$ nodes $v\in\mathbb{U}_d$ with $\hgt(v) = m$, 
we deduce from the above inequality that 
\begin{equation*}
    L(m,\varepsilon) \le (d+1)d^{m-1} \prob(\tilde{D}_1\tilde{D}_2\cdots \tilde{D}_{m-1} \ge \varepsilon),
\end{equation*}
where the $\tilde{D}_i$ are \textsc{iid} $\mathrm{Beta}(\sfrac{1}{(d-1)},1)$-distributed random variables. By applying Markov's inequality to the squares, and then Corollary~\ref{cor:beta_moment}, we obtain that
\[
    L(m,\varepsilon)\leq (d+1)d^{m-1}\varepsilon^{-2}\E[D_{(1,1)}^2]^{m-1} = (d+1)\varepsilon^{-2}\big(\tfrac{d}{2d-1}\big)^{m-1}\,.
\]
The desired inequality follows since $d\geq 2$.
\end{proof}

\subsection{Number of competitors in small subtrees}
Let $t$ be an arbitrary $d$-ary plane tree. The goal of this section is to uniformly bound the number of vertices of $t$ that are selected by the algorithm before the actual root, but that belong to $w*\theta_w t$ for some $w \in t$ for which $|\theta_w t|/|t|$ is small. In other words, we want to complete step \ref{3_key_step} from the proof of Theorem~\ref{thm:dary_weak} in Section~\ref{sect: sketch of proof}. Our strategy is to translate the problem into the framework presented in Section~\ref{sec:flow} by encoding all the proportions $|\theta_w t|/|t|$ into a single preflow $p:\bbU\to[0,1]$. It will then turn out that the number of vertices of $w*\theta_wt$ that are more central than the root in $t$ can be related to the number $N_x(p)=|\{u\in\bbU\, :\, x\prod_{\varnothing\prec v\preceq u}\tfrac{p(v)}{2}\geq 1\}|$, defined in  (\ref{def_E_and_N}), for some $x>0$. This motivates the following lemma.

\begin{lemma}\label{flow binary}
There is a constant $c_d>0$ such that for any $d$-ary preflow $p$ and any $x \geq 1$, it holds that $N_x(p)\leq \exp(c_d+c_d\sqrt{\log x})$.
\end{lemma}
\begin{proof}
Set $\alpha=d^{1/(d-1)}$, and define a function $\gamma=\gamma_{\alpha}:\bbU \to [0,1]$ by $\gamma(u_1,\ldots,u_h)=\alpha^{-\sum_{i=1}^h (u_i-1)}$ as in Proposition~\ref{deter_geo_flow}. Without loss of generality, we can assume that $p(u*1)\geq p(u*2)\geq\ldots\geq p(u*d)$ for any $u \in \mathbb{U}$. We claim that $p(u)\leq \gamma(u)$ for any $u \in \mathbb{U}$. To see this, first note that if $k\geq d+1$ then $p(u*k)=0\leq \gamma(u*k)$. If $k\in[d]$ then we have 
\[
p(u*k)\leq \frac{1}{k}\sum_{j=1}^k p(u*j)\leq \frac{p(u)}{k}\, ,
\]
since $p$ is a preflow. The choice of $\alpha$ yields that $\frac{1}{k} \le \alpha^{1-k}$ for $k \in [d]$; to see this note that by taking logs it suffices to show that $\sfrac{\log s}{(s-1)}$ is decreasing in $s$ for $s>1$, which is straightforward. 
Thus, $p(u*k) \le \alpha^{1-k}p(u)$, and the claim then follows by induction.  \smallskip

Since $p \le \gamma$, we deduce that for any $d$-ary preflow $p$, we have $N_x(p)\leq N_x(\gamma)$, 
from which we conclude the proof thanks to Proposition~\ref{deter_geo_flow}.
\end{proof}

The next proposition is the key bound on the number of nodes that are more central than the root, lying in a given small subtree. Recall from (\ref{PHI_def}) that $\Phi(t)=\tfrac{\varphi_t(\varnothing)}{\min_{u\in t}\varphi_t(u)}$.

\begin{proposition}\label{deterministicBound}
Let $c_d>0$ be the constant given in Lemma~\ref{flow binary}.
Then for any plane tree $t$ and any node $u\in t$, if $\theta_u t$ is $d$-ary and $|\theta_ut|\leq \frac{1}{3}|t|$ then $\big|\{ v \in \theta_ut: \varphi_{t}(\varnothing) \geq \varphi_{t}(u*v)\}\big| \leq \exp(c_d+c_d\sqrt{\log \Phi(t)})$.
\end{proposition}

\begin{proof}
Define a $d$-ary preflow by setting $p(v)=|\theta_{u*v} t|/|\theta_ut|$ for all $v\in\bbU$. For any $v\in\theta_ut$, using the expression (\ref{eq:phiu-ancestor}), we can then write 
\begin{align*}
    \frac{\varphi_{t}(u)}{\varphi_{t}({u*v})} = \prod_{\varnothing \prec w \preceq v}\frac{|\theta_{u*w}t|}{|t|-|\theta_{u*w} t|} =\prod_{\varnothing\prec w \preceq v}\frac{p(w)}{\frac{|t|}{|\theta_ut|}-p(w)}\,.
\end{align*}
Under the assumption that $|\theta_ut| \leq \frac{1}{3}|t|$, this provides the bound
\begin{align*}
    \frac{\varphi_{t}(u)}{\varphi_{t}(u*v)} \leq \prod_{\varnothing\prec w \preceq v}\frac{p(w)}{2}\, ,
\end{align*}
which, in turn, implies that
\[N_x(p)=\left|\left\{ v \in {\bbU} \, :\, \prod_{\varnothing\prec w\preceq v} \frac{p(w)}{2}\geq \frac{1}{x}\right\}\right| \geq \left|\left\{ v \in \theta_ut: \varphi_{t}(u)x \geq \varphi_{t}(u*v)\right\}\right|\]
for all $x>0$. Since $\Phi(t)\geq \varphi_t(\varnothing)/\varphi_t(u)$, taking $x=\Phi(t)$ in the above inequality, Lemma~\ref{flow binary} yields that
\begin{align*}
    \big|\{ v\in\theta_u t: \varphi_t(\varnothing)\geq \varphi_t(u*v)\}\big| 
    & \leq \big|\{v\in\theta_u t: \varphi_t(u)\Phi(t)\geq\varphi_t(u*v)\}\big|\\
    & \leq \exp(c_d+c_d\sqrt{\log\Phi(t)})
\end{align*}
as required.
\end{proof}

\subsection{Competitive ratio}

To derive a good bound of the number of competitors in a $(d+1)$-regular uniform attachment tree by using Proposition~\ref{deterministicBound}, we first need to control its competitive ratio. Hence, the goal of this section is to show the result below, which corresponds to step \ref{2_key_step} in the proof of Theorem~\ref{thm:dary_weak} in Section~\ref{sect: sketch of proof}.

\begin{proposition}\label{Distribution of Phi(T)}
There exist constants $c, C > 0$ such that the following holds. Fix $d \ge 2$, and for $n\geq 1$ let $T_n \sim \mathrm{UA}_{d+1}(n)$. Then $\limsup_{n\to\infty}\prob(\Phi(T_n) \geq x) \leq C x^{-c}$ for all $x \geq 1$. 
\end{proposition} 

The idea of the proof is to combine the deterministic upper bound given by Lemma~\ref{Upper bound Phi} for the competitive ratio with the asymptotic description of the proportions discussed in Section~\ref{presentationAndAsymptotics}. It turns out that for the bound we obtain in this way, the upper tail behaviour can be controlled using the tail bound from the following general lemma.

\begin{lemma}\label{Bound x}
For all $\beta \in (0,1)$ and $b \ge 1$ there exist $C,c>0$ such that the following holds. Let $(V_i)_{i \ge 1}$ be independent random variables supported by $[0,1]$ such that $\E[(1-V_i)^{-\beta}]\leq b$ for all $i \ge 1$. Let $G=\inf \{i\geq 1: V_1 \cdots V_i\leq 1/2\}$ and set
    \begin{align}
        X = \prod_{i = 1}^{G-1}\frac{1}{1 - V_1\cdot\ldots\cdot V_i}.
    \end{align}
Then $\prob(X \ge x) \le Cx^{-c}$ for all $x \ge 1$.
\end{lemma}
\begin{proof}
First, Markov's inequality yields that for all $i,m\geq 1$,
\[\prob\big(V_i\geq 2^{-1/m}\big)=\prob\big((1-V_i)^{-\beta}\geq (1-2^{-1/m})^{-\beta}\big)\leq b(1-2^{-1/m})^{\beta}.\]
Thus, there exists an integer $m=m(\beta,b)\geq 1$ such that $\prob(V_i\geq 2^{-1/m})\leq 1/2$ for all $i\geq 1$. Next, for any integer $g \ge 1$ we have 
\begin{align}\label{Bound on X}
    \prob(X\geq x)\leq \prob(G\geq mg+1)+\prob \left(\prod_{i=1}^{mg}\frac{1}{1-V_i}\geq x\right)
\end{align}
because $0\leq V_j\leq 1$ for any $j\geq 1$. Since $(V_i)_{i \ge 1}$ are independent, 
\[
\E \left[\left(\prod_{i=1}^{mg}\frac{1}{1-V_i}\right)^{\beta}\right] = \prod_{i=1}^{mg} \E\left[(1-V_i)^{-\beta}\right] \leq b^{mg},
\]
so Markov's inequality provides the bound
\begin{align}\label{Markov}
    \prob \left(\prod_{i=1}^{mg}\frac{1}{1-V_i}\geq x\right)\leq \frac{b^{mg}}{x^{\beta}}.
\end{align}
By the definition of $G$, if there exist $m$ distinct indices $i_1<\ldots<i_m$ such that $V_{i_1},\ldots,V_{i_m}$ are all less than $2^{-1/m}$, then $G\leq i_m$. Therefore, we can stochastically bound $G$ above by the sum of $m$ independent  geometric random variables $Y_1,\ldots,Y_m$ with common parameter $\inf_{i\geq 1} \prob(V_i\leq 2^{-1/m})$, which we recall to be at least $1/2$ by our choice of $m$. It follows that $\prob(G\geq mg+1)\leq m\prob(Y_1\geq g+1)\leq m /2^g$ for any $g\geq 0$. Then, we can rewrite (\ref{Bound on X}) by using (\ref{Markov}) and the previous inequality to obtain
\begin{align*}
    \prob(X\geq x)\leq \frac{m}{2^g}+ \frac{b^{mg}}{x^{\beta}}. 
\end{align*}
Finally, choose $c\in(0,1)$ small enough to have that $\beta- cm \log_2 b\geq c$, and let $g=\lfloor c\log_2 x\rfloor$. We then get $\prob(X\geq x)\leq (2m+1)x^{-c}$ as desired.
\end{proof}

Since we wish to apply Lemma~\ref{Bound x} to prove Proposition~\ref{Distribution of Phi(T)}, we state and prove a technical estimate which will ensure that the assumptions of Lemma~\ref{Bound x} are satisfied in the relevant situation. In what follows, the random variables $(D_u)_{u \in \bbU}$ are as in \eqref{def_P-U_binary}.

\begin{proposition}
\label{moment_max_-1/2}
Let $V_\varnothing=\max_{i\in[d+1]} D_i$ and $V_v=\max_{i \in [d]} D_{v*i}$ for all $v\in\bbU_d\setminus\{\varnothing\}$. Then, it holds that $\E[(1-V_u)^{-1/2}]\leq 7\sqrt{2}$ for all $u\in\bbU_d$.
\end{proposition}

\begin{proof}
Set $k=d+1$ if $u=\varnothing$ or $k=d$ otherwise, so that $V_u$ is the maximum of the components of a $\mathrm{Dir}_k(\tfrac{1}{d-1})$-distributed random vector by Proposition~\ref{Independence U binary}. The components of a symmetric Dirichlet random vector are non-negative and sum to $1$, so at most one of them is greater than $1/2$. Since they are identically distributed, we have 
\begin{align*} 
\E\big[(1-V_u)^{-1/2}\big]&\leq \sqrt{2}+k\E\big[\I{D_{u*1}>1/2}(1-D_{u*1})^{-1/2}\big]\\
&\leq \sqrt{2}+(d+1)\E\big[\I{D_{u*1}>1/2}(1-D_{u*1})^{-1/2}\big],
\end{align*}
where we recall from Section~\ref{sub:beta} that $D_{u*1}$ is $\mathrm{Beta}(\tfrac{1}{d-1},\tfrac{k-1}{d-1})$-distributed. \smallskip

As we have $\tfrac{k-1}{d-1}\ge 1$, the stochastic domination relation \eqref{eq:betadomination} implies that $D_{u*1} \preceq_{\mathrm{st}} Y\sim \mathrm{Beta}(\tfrac{1}{d-1},1)$. Since the function $x\in(0,1)\mapsto \I{x>1/2}(1-x)^{-1/2}$ is non-decreasing and non-negative, and the normalizing constant in the $\mathrm{Beta}(\tfrac{1}{d-1},1)$ distribution's density is $\Gamma(\tfrac1{d-1})\Gamma(1)/\Gamma(1+\tfrac{1}{d-1})=d-1$, 
it follows that
\begin{align*}
\E\big[(1-V_u)^{-1/2}\big]&\leq \sqrt{2}+(d+1)\E\big[\I{Y>1/2}(1-Y)^{-1/2}\big]\\
&=\sqrt{2}+\frac{d+1}{d-1}\int_{1/2}^1 (1-x)^{-1/2} x^{1/(d-1)-1}\, \dd x\, .
\end{align*}
Next, we use that $0\geq 1/(d-1)-1\geq -1$ to write
\[\E\big[(1-V_u)^{-1/2}\big]\leq \sqrt{2}+\frac{2(d+1)}{d-1}\int_{1/2}^1(1-x)^{-1/2}\, \dd x= \sqrt{2}\cdot \frac{3d+1}{d-1}\, .\]
The desired result follows since $d\geq 2$.
\end{proof}

\begin{proof}[Proof of Proposition~\ref{Distribution of Phi(T)}]
Using the random variables $(P_u)_{u \in \bbU}$ introduced in Section~\ref{presentationAndAsymptotics},  define a random path $(u_i)_{i\geq 0}$ in $\bbU_d$ as follows. 
For each node $u \in \bbU_d$, let $j(u)= \arg\max(D_{u*i}, i \ge 1)$, so that $P_{u*j(u)}=\max_{i \ge 1}P_{u*i}=P_uD_{u*j(u)}$. By Proposition~\ref{Independence U binary}, the random variables $(D_{u*i})_{i \ge 1}$ almost surely have a unique maximum so $j(u)$ is well-defined almost surely, and almost surely $j(\varnothing) \in [d+1]$ and $j(u) \in [d]$ for $u \ne \varnothing$. Observe that, in the notation of Proposition~\ref{moment_max_-1/2}, we have $V_u=D_{u*j(u)}$ for all $u \in \bbU_d$. 
\smallskip

Let $u_0=\varnothing$ and, inductively, for $i \ge 1$, given $u_i$ set $u_{i+1}=u_i*j_{u_i}$. We then have  $P_{u_{i}}=P_{u_{i-1}}D_{u_{i}}=\prod_{j=1}^i D_{u_j}$ for $i \ge 1$. By Proposition~\ref{Independence U binary}, the random pairs $(D_{u*j(u)};j(u))_{u\in\bbU_d}$ are independent
and have the same distribution except for $(D_{j_\varnothing};j(\varnothing))$, which implies that the random variables $(D_{u_i})_{i \ge 1}$ are independent and the random variables $(D_{u_i})_{i \ge 2}$ are identically distributed. Moreover, each $D_{u_i}$ has a continuous distribution and support $[0,1]$, so writing $G=\inf\{k \ge 1 :P_{u_k} \le 1/2\}$, it follows that $G$ is almost surely finite.
It also follows that $\prob(\exists k \in \N_1: \prod_{i \in [k]} D_{u_i}=1/2)=0$, so almost surely $P_{u_{G}} < 1/2$.
Since $|\theta_{u_i} T_n|/|T_n| \to P_{u_i}$ almost surely for all $i \ge 1$, it follows that there exists a random variable $n_0 \in \N_1$ such that almost surely, for all $n \ge n_0$, $|\theta_{u_i} T_n| > \tfrac12 |T_n|$ for all $i \in [G-1]$ and $|\theta_{u_{G-1}*j} T_n| < \tfrac12 |T_n|$ for all $j \in [d]$. It follows that for $n \ge n_0$,  $\varnothing = u_0 \preceq \dots \preceq u_{G-1}$
is the sequence of nodes described in Lemma~\ref{Upper bound Phi}. The conclusion of that lemma is that almost surely, for $n \ge n_0$, 
\[
\Phi(T_n) \le \prod_{i = 1}^{G-1} \frac{1}{1 - |\theta_{u_i}T_n|/|T_n|}\,.
\]
Since $|\theta_{u_i}T_n|/|T_n| \convas P_{u_i}$, it follows that $\limsup_{n \to \infty} \Phi(T_n) \le \prod_{i=1}^{G-1} (1-P_{u_i})^{-1}$ almost surely, so  
\[
\limsup_{n \to \infty} \prob(\Phi(T_n) \ge x) \le \prob\left(\prod_{i=1}^{G-1} \frac{1}{1-P_{u_i}} \ge x \right)
= 
\prob\left(\prod_{i=1}^{G-1} \frac{1}{1-D_{u_1}\cdot \ldots \cdot D_{u_i}} \ge x \right)
.  
\]
The random variables $(D_{u_i})_{i \ge 1}=(V_{u_{i-1}})_{i \ge 1}$ are independent as observed above, and by Proposition~\ref{moment_max_-1/2} we have $\E[(1-D_{u_i})^{-1/2}] \le 7\sqrt{2}$ for all $i \ge 1$. The result then follows from the above displayed equation and the bound of Lemma~\ref{Bound x}.
\end{proof}

\subsection{The dependence on \texorpdfstring{$d$}{d} in Theorem~\ref{thm:dary_weak}}
\label{rem:dependence_d}

By examining the proof given in Section~\ref{sect: sketch of proof}, we see that the constant $c_d^*$ appearing inside the exponential in Theorem~\ref{thm:dary_weak} depends on $d$ exclusively via the constant $c_d$ from Lemma~\ref{flow binary} and Proposition~\ref{deterministicBound}. However, one cannot remove the dependence in $d$ in Lemma~\ref{flow binary}; it is straightforward to show that if $p(\varnothing)=1$ and $p(u*i)=\I{i\leq d}\frac{p(u)}{d}$ for all $u\in\bbU$ and $i\geq 1$, then $p$ is a $d$-ary flow satisfying $\log N_x(p)\sim \sqrt{2(\log d)\log x}$ as $x\to\infty$. Thus, deterministic arguments are not enough to uniformly control the optimal constants $c_d^*$, or to bound from below the minimum size ensuring a given precision of the root-finding algorithm for random trees with unbounded degrees. 
Therefore, in Section~\ref{sec:ua}, 
when proving Theorems~\ref{thm:mainua} and~\ref{thm:dary}, we are forced to rely more heavily on the distributional properties of the models.

\section{Uniform attachment trees}\label{sec:ua}
The bulk of this section is devoted to the proof of Theorem~\ref{thm:mainua}. Since the proof of Theorem~\ref{thm:dary} is very similar to that of Theorem~\ref{thm:mainua}, our argument for it, which appears in Section~\ref{sec:dependence_d}, is much more succinct.

\subsection{Presentation and asymptotics of the model}
\label{presentationAndAsymptoticsUA}

We start by reformulating the uniform attachment model for plane trees using the Ulam--Harris formalism. Recall from Definition~\ref{def:plane_tree} that if $t$ is a plane tree and $u\in t$, then $\rk_u(t)$ stands for the number of children of $u$ in $t$. We let $T_1=\{\varnothing\}$ and for any $n\geq 1$, $T_{n+1}=T_n \cup \{V_n*(\rk_{V_n}(T_n)+1)\}$ where $V_n$ is uniformly random on $T_n$ conditionally given $(T_1,\ldots,T_n)$. By construction, the graph-isomorphism class of $T_n$ has the same law as that of $\mathrm{UA}(n)$-distributed tree presented in the introduction. Throughout this section, we slightly abuse notation by writing $\hat{T}_n\edi \mathrm{UA}(n)$ to mean that $\hat{T}_n$ is a random \emph{plane tree} with the same law as that of the plane tree $T_n$ constructed in this paragraph. This in particular allows us to write $T_n\edi\mathrm{UA}(n)$.
\smallskip

It is well-known, and straightforwardly proved by first and second moment computations, that in a uniform attachment tree the degree of any fixed node tends to $\infty$ as $n \to \infty$. For $u \in \bbU$ writing $\tau=\tau_u = \inf\{k \ge 1: u \in T_k\}$, so that $T_\tau=T_{\tau-1}\cup \{u\}$, it follows that $\tau_u$ is almost surely finite for all $u \in \bbU$; we will use this below.
\smallskip

Note that $|T_n|=n$ for any $n\geq 1$. As we did in Section~\ref{presentationAndAsymptotics} for the $d$-regular uniform attachment model, we describe the asymptotic behaviour of the proportions of the subtrees of $T_n$ via a comparison with certain nested P\'olya urns.

First, we show that for any $u \in \mathbb{U}$ and any integer $i\geq 1$, the two limits
\begin{equation}
\label{def_P_U_general-degree}
P_u = \lim_{n \to \infty} \frac{|\theta_uT_n|}{n}\quad\text{ and }\quad U_{u*i} = \lim_{n \to \infty}\I{u*i\in T_n}\frac{|\theta_{u*i}T_n|}{\sum_{j\geq i} |\theta_{u*j}T_n|}
\end{equation}
exist almost surely.
To do this, fix $u \in \bbU\setminus \{\varnothing\}$ and write $\tau=\tau_u$. Colour $u$ in black and all the vertices of $T_{\tau-1}$ in red, then colour each node added after time $\tau$ in the same colour as its parent. Hence, conditionally given $T_{\tau}$, $(|\theta_uT_{\tau+n}|)_{n \ge 0}$ evolves as the number of black balls in a standard Pólya urn that initially contains $1$ black ball and $\tau-1$ red balls, which classically implies that $\tfrac{1}{\tau+n}|\theta_u T_{\tau+n}|$ converges towards a nonzero limit almost surely. This shows that $P_u$ is well-defined and that a.s.\ $P_u > 0$.
It follows that the limit $U_{u*i}$ also exists almost surely, and that $U_{u*i}\aseq P_{u*i}(P_u-\sum_{j=1}^{i-1} P_{u*j})^{-1}$. From there, an induction on $i\geq 1$ yields the following recursive relation:
\begin{equation}
\label{formula_P2}
\forall u\in\bbU,\forall i\geq 1,\quad P_{u*i} = P_u\cdot U_{u*i}\prod_{j=1}^{i-1}(1-U_{u*j})\,.
\end{equation}
For example, we get $P_1 = U_1$, $P_2=(1-U_1)U_2$, $P_3 = (1-U_1)(1-U_2)U_3$, and $P_{(1,1,2)} = U_1U_{(1,1)}(1-U_{(1,1,1)})U_{(1,1,2)}$. By induction on $u$, (\ref{formula_P2}) allows us to express all the asymptotic proportions $(P_u)_{u\in\bbU}$ in terms of the uniforms $(U_v)_{v\in\bbU\setminus\{\varnothing\}}$ as follows:
\begin{equation}
\label{formula_P1}
\text{if}\quad u = (u_1,u_2,\ldots,u_k),\quad\text{ then }\quad P_u = \prod_{i=1}^k\left(U_{(u_1,\ldots,u_i)}\prod_{j=1}^{u_i - 1}(1-U_{(u_1,\ldots,u_{i-1},j)})\right).
\end{equation}

Similarly as in Section~\ref{presentationAndAsymptotics}, we can describe the joint law of the factors $(U_u)_{u\in\bbU\setminus\{\varnothing\}}$.

\begin{proposition}\label{Independence U in UA}
The random variables $(U_u)_{u \in \bbU\setminus\{\varnothing\}}$ are Uniform$[0,1]$-distributed and independent .
\end{proposition}
\begin{proof}
Let $(X_n)_{n\geq 2}$ denote the number of black balls at time $n$ in a P\'olya urn that contains $1$ black ball and $1$ red ball at time $2$, so that $X_2=1$. Let also $(T_n')_{n\geq 1}$ and $(T_n'')_{n\geq 1}$ be two sequences of plane trees that both follow the uniform attachment model, i.e., they have the same law as $(T_n)_{n\geq 1}$. We assume that $(X_n)_{n\geq 2}$, $(T_n')_{n\geq 1}$, and $(T_n'')_{n\geq 1}$ are independent.
Then, for all $u \in \bbU$ and $i \ge 1$, let
\[
U_{u*i} = \lim_{n \to \infty}\I{u*i\in T_n'}\frac{|\theta_{u*i}T_n'|}{\sum_{j\geq i} |\theta_{u*j}T_n'|}
\quad\mbox{ and }\quad 
U_{u*i}'' = \lim_{n \to \infty}\I{u*i\in T_n''}\frac{|\theta_{u*i}T_n''|}{\sum_{j\geq i} |\theta_{u*j}T_n''|}
\]
be defined from $(T_n')$ and $(T_n'')$ as $U_u$ is defined from $(T_n)$ in (\ref{def_P_U_general-degree}). 

We now inductively construct another growing sequence $(\hat{T}_n)_{n\geq 1}$ of plane trees. Set $\hat{T}_1=\{\varnothing\}$ and $\hat{T}_2=\{\varnothing,1\}$. Then for any $n\geq 2$, define $\hat{T}_{n+1}$ as follows:
\begin{itemize}
\item[$\bullet$] if $X_{n+1}=X_n+1$ then $\hat{T}_{n+1}=\hat{T}_n\cup\{1*w'\}$, where $w'\in\bbU$ is the only node of $T_{X_n+1}'$ that is not in $T_{X_n}'$;
\item[$\bullet$] if $X_{n+1}=X_n$ then $\hat{T}_{n+1}=\hat{T}_n\cup\{(j+1)*w''\}$, where $j\geq 1$ and $w''\in\bbU$ are such that $j*w''$ is the only node of $T_{n-X_n+1}''$ that is not in $T_{n-X_n}''$.
\end{itemize}

For any $n\geq 2$, observe that by construction $|\theta_{1}\hat{T}_n|=X_n$ and $\theta_{1} \hat{T}_n=T_{X_n}'$, and that $\theta_{1+j}\hat{T}_n=\theta_{j} T_{n-X_n}''$ for all $j\geq 1$. This allows us to check by induction that $(\hat{T}_n)_{n\geq 1}$ follows the uniform attachment model, and thus we can assume without loss of generality that $\hat{T}_n=T_n$ for all $n\geq 1$. Then, for any $v\in\bbU$, we can write $\theta_{1*v}T_n=\theta_{v}\theta_1 T_n=\theta_v T_{X_n}'$ and similarly $\theta_{(1+j)*v}T_n=\theta_{j*v}T_{n-X_n}''$ for all $j\geq 1$. Since we know from standard properties of P\'olya urns that $X_n\to\infty$ and $n-X_n\to\infty$ almost surely, these identities and the expressions (\ref{def_P_U_general-degree}) entail that for all $u\in\bbU\setminus\{\varnothing\}$, $v\in\bbU$, and $j\geq 1$, we have
\[U_{1*u}=U_u',\quad U_{(1+j)*v}=U_{j*v}'',\quad U_1=\lim_{n\to\infty} \tfrac{1}{n}X_n\, .\]
By independence of $(X_n)$, $(T_n')$, and $(T_n'')$, it follows that the families $U_1$, $(U_{1*u})_{u\in\bbU\setminus\{\varnothing\}}$, and $(U_{(1+j)*v})_{j\geq 1,v\in\bbU}$ are independent. Since $(T_n)\stackrel{d}{=}(T_n')\stackrel{d}{=}(T_n'')$, it also follows that $(U_{1*u})_{u\in\bbU\setminus\{\varnothing\}}$ has the same law as $(U_u)_{u\in\bbU\setminus\{\varnothing\}}$, and $(U_{(1+j)*v})_{j\geq 1, v\in\bbU}$ has the same law as $(U_{j*v})_{j\geq 1,v\in\bbU}$. Moreover, it is a well-known fact about P\'olya urns that $U_1=\lim_{n \to \infty} X_n/n$ is Uniform$[0,1]$-distributed. An inductive argument concludes the proof.
\end{proof}

Define the \emph{weight} of $u=(u_1,\ldots,u_i)\in\mathbb{U}$ as the sum of its letters: namely, 
\begin{equation}
\label{weight_def}
\rw(u)=\sum_{i=1}^n u_i\, .
\end{equation}
This quantity exactly corresponds to the weight function defined in Section \ref{sect: sketch of proof}. Remark that $P_u$ is the product of exactly $\rw(u)$ independent Uniform$[0,1]$-distributed random variables.

\smallskip 
Building on Lemma \ref{Condition on size of subtrees}, we  next show that with high probability, there are only finitely many nodes that are more central than the root in a infinite plane tree generated by the uniform attachment model. This is because the weights $\rw(u)$ of such nodes must stay bounded. Moreover, thanks to this fact, we will avoid some technical issues by considering only finitely many convergences at a time.

\begin{lemma}\label{Size subtree} 
For any $n\geq 1$, let $T_n \sim \mathrm{UA}(n)$. 
Then for all integers $m\geq 1$ and for all $\varepsilon\in (0,1)$, it holds that
\[\limsup_{n\to\infty}\prob\big(\exists u\in\mathbb{U}\, :\, \rw(u)\geq m, |\theta_u T_n|\geq \varepsilon n\big)\leq \varepsilon^{-2}(2/3)^{m-1}. \]
\end{lemma}

\begin{proof}
We define $L(m,\varepsilon)$ to be the $\limsup$ in the statement of the lemma. Note that for $u \in \mathbb{U}$, if $\rw(u) \geq m$, then there exist a node $v$ and an integer $k \ge 1$ such that $\rw(v*k) = m$ and $v*k$ is a (weakly) older sibling of some (weak) ancestor of $u$ (that is, there are $w\in \mathbb{U}$ and an integer $i\geq k$ such that $v*i*w=u$). This gives that $\sum_{i=k}^\infty|\theta_{v*i}T_n| \ge |\theta_u T_n|$, hence
\begin{equation*}
    L(m,\varepsilon) \le \limsup_{n\to\infty}\prob\left(\exists v\in\mathbb{U},k\geq 1\, :\, \rw(v*k)=m, \sum_{i=k}^\infty|\theta_{v*i}T_n|\geq \varepsilon n\right)\,.
\end{equation*}
Since there are only finitely many nodes $u$ with $\rw(u)=m$, we obtain that
\begin{equation*}
    L(m,\varepsilon) \le \sum_{\substack{v\in\bbU, k\geq 1 \\ \rw(v*k)=m}} \limsup_{n\to\infty}\prob\left(\sum_{i=k}^\infty\frac{|\theta_{v*i}T_n|}{n}\geq \varepsilon \right)\,.
\end{equation*}
We can rewrite the sum inside the probability as $\frac{1}{n}(|\theta_vT_n|-1 -\sum_{i=1}^{k-1} |\theta_{v*i}T_n|)$.  Letting $n \to \infty$, this random variable converges almost surely to $P_v-\sum_{i=1}^{k-1} P_{v*i}$, which is equal to $P_v\prod_{i=1}^{k-1}(1-U_{v*i})$ using (\ref{formula_P2}) and a telescoping argument. Moreover, by (\ref{formula_P1}), note that $P_v\prod_{i=1}^{k-1}(1-U_{v*i})$ is a product of $ \rw(v)+k-1=\rw(v*k) - 1 = m-1$ independent Uniform$[0,1]$-distributed random variables. By the Portmanteau theorem we can thus rewrite the above inequality as
\begin{equation*}
    L(m,\varepsilon) \le \big|\{u\in\bbU\, :\, \rw(u)=m\}\big|\cdot \prob(U_1U_2\cdots U_{ m-1} \ge \varepsilon),
\end{equation*}
where the $U_i$ are independent Uniform$[0,1]$-distributed random variables. It is straightforward to verify by induction that there are exactly $2^{m-1}$ nodes $u$ with $\rw(u) = m$, so we can further rewrite the inequality as
\begin{equation*}
    L(m,\varepsilon) \le 2^{m-1 }\prob(U_1U_2\cdots U_{m-1} \ge \varepsilon) = 2^{m-1}\prob((U_1U_2\cdots U_{m-1})^2 \ge \varepsilon^2).
\end{equation*}
 Using Markov's inequality and independence, we obtain that
 \[L(m,\varepsilon)\leq \frac{2^{m-1}}{\varepsilon^2}\E[U_1^2]^{m-1} = \frac{2^{m-1}}{\varepsilon^2 3^{m-1}}\, .\qedhere\]
\end{proof}

\subsection{Number of competitors in small subtrees}

The random variables $(P_u)_{u \in \bbU}$ define a (random) flow on $\bbU$; however, it is somewhat complex to work with,
so we bound it from above by another, better-behaved family of random variables $(Q_u)_{u \in \bbU}$.
The upper bound will be useful as will enforce an exponential decay in the values of ``younger siblings'', which will be simpler than but similar to the expected behaviour of $(P_u)_{u \in\bbU}$. We then use the bound to prove the main result of the section, a probabilistic analogue of Lemma~\ref{flow binary} which establishes an exponential upper tail bound  for the random variable $N_x(P)=|E_x(P)|$, where we recall that $E_x(P)=\{u \in \bbU: x \prod_{\varnothing \prec v \preceq u} \tfrac{P_v}{2} \ge 1\}$.
\smallskip

We begin with a technical lemma; we will then immediately use the lemma to construct the family of random variables $(Q_u)_{u \in \bbU}$.
\begin{lemma}
\label{Unif}
Let $(U_i)_{i \ge 1}$ be independent Uniform$[0,1]$-distributed random variables. Then there exists a sequence $(V_i)_{i \ge 1}$ of independent Uniform$[1/2,1]$-distributed random variables and a random bijection $\sigma:\N_1 \to \N_1$ such that with $P_i=U_i\cdot \prod_{j=1}^{i-1}(1-U_j)$, then almost surely $P_{\sigma(i)} \le \tfrac12\prod_{j=1}^{i-2}V_j$ for all $i \ge 2$.
\end{lemma}
\begin{proof} 
First, let $K = \inf \big\{i \ge 1 : U_i > \frac{1}{2}\big\}$ which is almost surely finite. Then for $i \ge 1$ set
\begin{align*}
    V_i = \begin{cases}
        1 - U_i &\text{ if } i \le K-1, \\
        1 - \frac{U_{i+1}}{2} &\text{ if } i \ge K ,
    \end{cases}
\end{align*}
and define $\sigma:\N_1\to\N_1$ by
\begin{align*}
    \sigma(i) = \begin{cases}
        K &\text{ if } i=1, \\
        i-1 &\text{ if } 2 \le i \le K, \\
        i &\text{ if } i \ge K+1.
    \end{cases}
\end{align*}
For $2 \le i \le K$, it holds that $U_{i-1} \le \frac{1}{2}$ by construction, so
\begin{align*}
    P_{\sigma(i)} =P_{i-1}= U_{i-1}\prod_{j=1}^{i-2}(1-U_j) \le \frac{1}{2}\prod_{j=1}^{i-2}V_j\,.
\end{align*}
Also, for $i \ge K+1$, since $U_K > 1/2$, we have that
\begin{align*}
    P_{\sigma(i)}  = P_i &= U_i \left(\prod_{j=1}^{K-1}(1-U_j)\right)(1-U_K)\left(\prod_{j=K+1}^{i-1}(1-U_j)\right) \\
    & \leq \left(\prod_{j=1}^{K-1}V_j\right) \cdot \frac{1}{2}\cdot \left(\prod_{j=K+1}^{i-1} V_{j-1} \right)\,.
\end{align*}
Hence, for all $i \ge 2$, we indeed have that $P_{\sigma(i)} \leq \frac{1}{2}\prod_{j=1}^{i-2}V_j$ almost surely.
\smallskip

We now show that the $V_i$'s are independent and Uniform$[1/2,1]$-distributed. Let $i\geq 1$, let $f_1,\ldots,f_i:\R\to\R$ be bounded measurable functions, and let $W\sim \mathrm{Uniform}[1/2,1]$.  It suffices to show that $\E\left[\prod_{j=1}^i f_j(V_j)\right] = \prod_{j=1}^i \E[f_j(W)]$.  We write the left-hand side as a sum over every possible value of $K$ using indicators:
\begin{equation*}
    \E\left[\prod_{j=1}^i f_j(V_j)\right]\! = \sum_{k \geq 1} \E\left[\left(\prod_{j=1}^{k-1}\I{U_j\leq \frac{1}{2}}f_j(1-U_j)\right)\!\cdot\I{U_k>\frac{1}{2}}\!\cdot\left(\prod_{j=k}^if_j\left(1-\frac{U_{j+1}}{2}\right)\right)\right].
\end{equation*}
Notice that each factor in the product on the right is a function of a different random variable, so since the $(U_i)_{i\ge 1}$ are \textsc{iid}, we obtain
\begin{equation*}
    \E\bigg[\prod_{j=1}^i f_j(V_j)\bigg] \!= \sum_{k \geq 1} \bigg(\prod_{j=1}^{k-1}\E\bigg[\I{U_1\leq \frac{1}{2}}f_j(1-U_1)\bigg]\bigg)\!\cdot\prob\big(U_1>\tfrac{1}{2}\big)\cdot\!\bigg(\prod_{j=k}^i\E\bigg[ f_j(1-U_1/2)\bigg]\bigg).
\end{equation*}
Finally, we observe that $W$ has the same law as $1-U_1/2$, and that the conditional law of $1-U_1$ given that $U_1\leq 1/2$ is also Uniform$[1/2,1]$-distributed. This entails that $\E\left[\prod_{j=1}^i f_j(V_j)\right]= \sum_{k \ge 1} 2^{-k}\prod_{j=1}^i\E[f_j(W)]$ which implies the desired identity.
\end{proof}

\begin{corollary}\label{corr}
There exists a family $(V_u)_{u\in\bbU\setminus\{\varnothing\}}$ of \textsc{iid}~Uniform$[1/2,1]$-distributed random variables and a random bijection $\sigma:\bbU\to\bbU$ such that $\sigma(\varnothing)=\varnothing$ and $\sigma(\ola{u})=\ola{\sigma(u)}$ for all $u\in\bbU\setminus\{\varnothing\}$, and $P_{\sigma(u)}\leq Q_u$ for all $u\in\bbU$ almost surely, where $(Q_u)_{u\in\bbU}$ is inductively defined by
\begin{itemize}
    \item $Q_\varnothing=1$;
    \item $Q_{u*1}=Q_u$ for all $u\in\bbU$;
    \item $Q_{u*i}=Q_u\cdot \frac{1}{2}\prod_{j=1}^{i-2} V_{u*j}$ for all $i\geq 2$ and all $u\in\bbU$.
\end{itemize}
\end{corollary}
\begin{proof}
For $u \in \bbU$ write $\mathbf{U}_u=(U_{u*i})_{i \ge 1}$. 
From Proposition~\ref{Independence U in UA} we have that the random vectors $(\mathbf{U}_u)_{u \in \bbU}$ are independent, and within each vector the entries $(U_{u*i})_{i\geq 1}$ are \textsc{iid}~and Uniform$[0,1]$-distributed. For each $u \in \bbU$, Lemma~\ref{Unif} thus gives us a sequence $(V_{u*i}')_{i \ge 1}$ of independent Uniform$[1/2,1]$-distributed random variables and a bijection $\sigma_u:\N_1 \to \N_1$ such that $U_{u*\sigma_u(i)}\prod_{j=1}^{\sigma_u(i)-1}(1-U_{u*j}) \le \frac{1}{2}\cdot\prod_{j=1}^{i-2}V_{u*j}'$ for all $i\geq 2$ almost surely. 
Moreover, the independence of the vectors $(\mathbf{U}_u)_{u \in \bbU}$ implies that the sequences $(\sigma_u(i),V_{u*i}')_{i\geq 1}$ for $u\in\bbU$ can be chosen to be independent. 
\smallskip

Now, let us define $\sigma:\bbU \to \bbU$ inductively by setting $\sigma(\varnothing) = \varnothing$ and $\sigma(u*i) = \sigma(u) * \sigma_{\sigma(u)}(i)$ for all $u\in\bbU$ and $i\geq 1$, so that $\overleftarrow{\sigma(u)} = \sigma(\overleftarrow{u})$ for all $u \in \bbU\setminus\{\varnothing\}$. It is then easy to check that $\sigma$ is a bijection on $\bbU$ which preserves the height, i.e.~$\hgt(\sigma(u))=\hgt(u)$ for all $u\in\bbU$. Furthermore, for any $h\geq 0$, the restriction of $\sigma$ to $\{v\in\bbU:\hgt(v)\leq h\}$ is measurable with respect to the family $(\sigma_v(i),V_{v*i}'\, ;\, i\geq 1,v\in\bbU,\hgt(v)\leq h-1)$. Setting $V_{u*i}=V_{\sigma(u)*i}'$ for all $u\in\bbU$ and $i\geq 1$, it follows via induction on the height that the random variables $(V_u)_{u\in\bbU\setminus\{\varnothing\}}$ are independent Uniform$[1/2,1]$-distributed random variables as desired.
\smallskip

From (\ref{formula_P2}), for all $u\in\bbU$ and $i \ge 2$, we see that $P_{\sigma(u*1)}=P_{\sigma(u)*\sigma_{\sigma(u)}(1)}\leq P_{\sigma(u)}$ and that
\[P_{\sigma(u*i)} = P_{\sigma(u) * \sigma_{\sigma(u)}(i)} \le P_{\sigma(u)}\cdot \frac{1}{2}\cdot\prod_{j=1}^{i-2}V_{\sigma(u)*j}'=P_{\sigma(u)}\cdot \frac{1}{2}\cdot\prod_{j=1}^{i-2}V_{u*j}=P_{\sigma(u)}\frac{Q_{u*i}}{Q_u}\, .\] A straightforward induction then yields that $P_{\sigma(u)}\leq Q_u$ for all $u\in\bbU$ almost surely.
\end{proof}

By contrast with regular trees, in the uniform attachment case, we cannot invoke a deterministic bound such as Proposition~\ref{deterministicBound}. Nonetheless, we can provide a probabilistic replacement of Proposition~\ref{deterministicBound} that relies on the distribution of the family $P=(P_u)_{u\in\bbU}$ by using Corollary~\ref{corr}. This result represents the main goal of the present section and will be crucial to prove Theorem~\ref{thm:mainua} later. Recall from (\ref{def_E_and_N}) that for $f:\bbU\to [0,\infty)$ and $x>0$, we write $E_x(f)=\big\{u\in \mathbb{U}: x\prod_{\varnothing\prec v\preceq u}\frac{f(v)}{2}\geq 1\big\}$ and  $N_x(f)=|E_x(f)|$.

\begin{proposition}\label{main}
There exists a constant $c>0$ such that for any $x \geq 1$ and $y\geq 1$,
\[\prob\big(N_x(P)\geq y \exp(c+c\sqrt{\log x}) \big)\leq 2e^{-y}.\]
\end{proposition}

We now introduce some notation and prove a technical lemma that will be used to show Proposition~\ref{main}.  
In what follows, we use the function $\gamma=\gamma_{4/3}$ defined in Proposition~\ref{deter_geo_flow}; recall this function is given by $\gamma(u)=(4/3)^{-\sum_{i=1}^h (u_i-1)}$ for $u=(u_1,\ldots,u_h) \in \bbU$. Also, for $u = (u_1, \dots, u_h) \in \mathbb{U}$, let $\rr(u) = \sum_{i = 1}^{l}\I{u_i \geq 2}$ be the number of nodes on the ancestral path of $u$ which are not oldest children.

\begin{lemma}\label{Bound r(u)}
There is a constant $c>0$ such that $\sup_{u\in E_x(\gamma)}\rr(u)\leq c\sqrt{\log x}$ for all $x\geq 1$.
\end{lemma}

\begin{proof}
Fix $u=(u_1,\ldots,u_h) \in \bbU$. Notice that if $i$ is the $k$'th index for which $u_i \ge 2$, then $h+1-i \ge \rr(u)+1-k$. This entails that $\sum_{i=1}^h (h+1-i)(u_i-1) \ge \sum_{k=1}^{\rr(u)} (\rr(u)+1-k)=\tfrac{1}{2}\rr(u)(\rr(u)+1)$. Therefore, it follows from the identity (\ref{product_deter_geo_flow}) with $\alpha=4/3$  that 
\[\prod_{\varnothing\prec v \preceq u} \frac{\gamma(v)}{2} 
= \frac{1}{2^h}\prod_{i=1}^{h} \left(\frac34\right)^{(h+1-i)(u_i-1)} \le \left(\frac34\right)^{\rr(u)(\rr(u)+1)/2}\, ,\]
so if $u \in E_x(\gamma)$ then $(4/3)^{\rr(u)^2/2} \le x$, and thus $\rr(u) \le \sqrt{2\log_{4/3} x}$. The result follows.
\end{proof}

\begin{proof}[Proof of Proposition~\ref{main}]
Our strategy for bounding $N_x(P)$ is to compare it with $N_x(\gamma)$, where $\gamma=\gamma_{4/3}$ is as above and in Proposition~\ref{deter_geo_flow}. Let $Q=(Q_u)_{u\in\bbU}$ and $\sigma:\bbU\to\bbU$ be as in Corollary~\ref{corr}.
From the identity $\sigma(\overleftarrow{u})=\overleftarrow{\sigma(u)}$, we observe that the ancestors of $\sigma(u)$ are the nodes $\sigma(v)$ with $v\preceq u$. Since $P_{\sigma(u)} \le Q_u$, it follows that if $\sigma(u)\in E_x(P)$ then $u\in E_x(Q)$, so $N_x(P)\leq N_x(Q)$ almost surely.
\smallskip

For $v\in\bbU$, let $\tau_{v,1}=0$ and, for all $i\geq2$, let  $\tau_{v,i}=\inf\{j>\tau_{v,i-1}:V_{v*j}\leq \frac{3}{4}\}$, where $(V_{u})_{u \in \bbU\setminus\{\varnothing\}}$ are given by Corollary~\ref{corr}. 
Note that the increments $\tau_{v,i}-\tau_{v,i-1}$ are Geometric$(1/2)$-distributed.
Moreover, it follows from the independence of the random variables 
$(V_u)_{u \in \bbU\setminus\{\varnothing\}}$ that 
the increments $(\tau_{v,i}-\tau_{v,i-1}\, ;\, v \in\bbU,i \ge 2)$ are also independent. Now define random variables $(\Gamma_{u})_{u \in \bbU}$ inductively by setting $\Gamma_\varnothing=1$ and, for $v \in \bbU$,
\begin{itemize}
    \item $\Gamma_{v*1}=\Gamma_v$; and
    \item for each $i,j\geq 2$, $\Gamma_{v*i} = \Gamma_v\left(\frac{3}{4}\right)^{j-1}$ if and only if $\tau_{v,j-1} +1< i \leq \tau_{v,j}+1$. 
\end{itemize}
We have $Q_{v*i}\leq Q_v\cdot (3/4)^{j-1}$ for all $i>\tau_{v,j-1}+1$: indeed, $j-2$ factors $3/4$ come from the fact that $V_{v*\tau_{v,k}}\leq 3/4$ for each $2\leq k\leq j-1$, and the last factor $3/4$ comes from the $1/2$ involved in the recursive expression of $Q_{v*i}$. By construction of $Q$ and $\Gamma=(\Gamma_u)_{u\in\bbU}$, it then follows that $Q_u \leq \Gamma_u$ for all $u\in\bbU$, so $N_x(Q)\leq N_x(\Gamma)$.

\smallskip

\begin{figure}
    \begin{tikzpicture}[
  dot/.style = {circle, fill=black, inner sep=2pt, draw=none},
  square/.style = {regular polygon, regular polygon sides=4, fill=black, inner sep=2pt, draw=none},
  triangle/.style = {regular polygon, regular polygon sides=3, fill=black, inner sep=1.5pt, draw=none},
  level distance = 15mm,
  grow = up,
  level 1/.style = {sibling distance=20mm},
  level 2/.style = {sibling distance=3mm}
]

\begin{scope}[xshift=-20mm] 
\node[dot] (root1) {}
  child[edge from parent/.style={draw=red!70!black, line width=0.8pt}] {node[dot] (a) {}
    child[edge from parent/.style={draw=red!70!black, line width=0.8pt}] {node[triangle] (a1) {}}
    child[edge from parent/.style={draw=RoyalBlue!60!white, line width=1.8pt}] {node[dot] (a2) {}}
    child[edge from parent/.style={draw=black, line width=2.2pt}] {node[dot] (a3) {}}
  }
  child[edge from parent/.style={draw=RoyalBlue!60!white, line width=1.8pt}] {node[dot] (b) {}
    child[edge from parent/.style={draw=red!70!black, line width=0.8pt}] {node[dot] (b1) {}}
    child[edge from parent/.style={draw=RoyalBlue!60!white, line width=1.8pt}] {node[dot] (b2) {}}
    child[edge from parent/.style={draw=black, line width=2.2pt}] {node[square] (b3) {}}  
  }
  child[edge from parent/.style={draw=black, line width=2.2pt}] {node[dot] (c) {}
    child[edge from parent/.style={draw=red!70!black, line width=0.8pt}] {node[dot] (c1) {}} 
    child[edge from parent/.style={draw=RoyalBlue!60!white, line width=1.8pt}] {node[dot] (c2) {}}
    child[edge from parent/.style={draw=black, line width=2.2pt}] {node[dot] (c3) {}}
  };

\node[above=35mm of root1] {The geometric flow $\gamma(u)$};
\node[draw=none, fill=none, right=1mm] at (a) {$\cdots$};
\node[draw=none, fill=none, right=1mm] at (a1) {$\cdots$};
\node[draw=none, fill=none, right=1mm] at (b1) {$\cdots$};
\node[draw=none, fill=none, right=1mm] at (c1) {$\cdots$};

\node[draw=none, fill=none, above=1mm] at (b3) {a};
\node[draw=none, fill=none, above=1mm] at (a1) {b};
\end{scope}

\begin{scope}[shift={(0,-45mm)}] 
\node (root2) {};
\node[dot] {}
  child[edge from parent/.style={draw=red!70!black, line width=0.8pt}] {node[dot] (a) {}
    child[edge from parent/.style={draw=red!70!black, line width=0.8pt}] {node[triangle] (a1) {}}
    child[edge from parent/.style={draw=RoyalBlue!60!white, line width=1.8pt}] {node[dot] (a2) {}}
    child[edge from parent/.style={draw=RoyalBlue!60!white, line width=1.8pt}] {node[dot] (a3) {}}
    child[edge from parent/.style={draw=RoyalBlue!60!white, line width=1.8pt}] {node[dot] (a4) {}}
    child[edge from parent/.style={draw=black, line width=2.2pt}] {node[dot] (a5) {}}
  }
  child[edge from parent/.style={draw=red!70!black, line width=0.8pt}] {node[dot] (b) {}
    child[edge from parent/.style={draw=red!70!black, line width=0.8pt}] {node[triangle] (b1) {}}
    child[edge from parent/.style={draw=red!70!black, line width=0.8pt}] {node[triangle] (b2) {}}
    child[edge from parent/.style={draw=red!70!black, line width=0.8pt}] {node[triangle] (b3) {}} 
    child[edge from parent/.style={draw=RoyalBlue!60!white, line width=1.8pt}] {node[dot] (b4) {}}
    child[edge from parent/.style={draw=black, line width=2.2pt}] {node[dot] (b5) {}}
  }
  child[edge from parent/.style={draw=RoyalBlue!60!white, line width=1.8pt}] {node[dot] (c) {}
    child[edge from parent/.style={draw=RoyalBlue!60!white, line width=1.8pt}] {node[dot] (c1) {}} 
    child[edge from parent/.style={draw=RoyalBlue!60!white, line width=1.8pt}] {node[dot] (c2) {}}
    child[edge from parent/.style={draw=black, line width=2.2pt}] {node[square] (c3) {}}
  }
  child[edge from parent/.style={draw=RoyalBlue!60!white, line width=1.8pt}] {node[dot] (d) {}
    child[edge from parent/.style={draw=RoyalBlue!60!white, line width=1.8pt}] {node[dot] (d1) {}}
    child[edge from parent/.style={draw=RoyalBlue!60!white, line width=1.8pt}] {node[dot] (d2) {}}
    child[edge from parent/.style={draw=RoyalBlue!60!white, line width=1.8pt}] {node[dot] (d3) {}}
    child[edge from parent/.style={draw=black, line width=2.2pt}] {node[square] (d5) {}}
  }
  child[edge from parent/.style={draw=RoyalBlue!60!white, line width=1.8pt}] {node[dot] (e) {}
    child[edge from parent/.style={draw=RoyalBlue!60!white, line width=1.8pt}] {node[dot] (e1) {}}
    child[edge from parent/.style={draw=black, line width=2.2pt}] {node[square] (e2) {}}
  }
  child[edge from parent/.style={draw=black, line width=2.2pt}] {node[dot] (f) {}
    child[edge from parent/.style={draw=RoyalBlue!60!white, line width=1.8pt}] {node[dot] (f1) {}}
    child[edge from parent/.style={draw=RoyalBlue!60!white, line width=1.8pt}] {node[dot] (f2) {}}
    child[edge from parent/.style={draw=black, line width=2.2pt}] {node[dot] (f3) {}}
  };
  
\node[above=32mm of root2] {The random flow $\Gamma(u)$};
\node[draw=none, fill=none, right=1mm] at (a) {$\cdots$};
\node[draw=none, fill=none, right=1mm] at (a1) {$\cdots$};
\node[draw=none, fill=none, right=1mm] at (b1) {$\cdots$};
\node[draw=none, fill=none, right=1mm] at (c1) {$\cdots$};
\node[draw=none, fill=none, right=1mm] at (d1) {$\cdots$};
\node[draw=none, fill=none, right=1mm] at (e1) {$\cdots$};
\node[draw=none, fill=none, right=1mm] at (f1) {$\cdots$};
\end{scope}

\begin{scope}[xshift=30mm, yshift=25mm] 
\node[label=right:{Equal to parent}] (l1) {};
\node[below=5mm of l1, label=right:{3/4 of parent}] (l2) {};
\node[below=5mm of l2, label=right:{$(3/4)^2$ of parent}] (l3) {};
\draw[black, line width=2.2pt] (l1.west) -- ++(-10mm,0);
\draw[RoyalBlue!60!white, line width=1.8pt] (l2.west) -- ++(-10mm,0);
\draw[red!70!black, line width=0.8pt] (l3.west) -- ++(-10mm,0);
\end{scope}

\end{tikzpicture}
    \caption{Illustration of the map $\chi$ and of the variables $Z_u$. The colour and thickness of each edge indicate the ratios of the values taken by the flow between the upper-end and the lower-end of the edge. Two vertices $a=(2,1)$ and $b=(3,3)$ are marked in the tree at the top. In the tree at the bottom, three nodes have $a$ as their $\chi$-image (marked with squares), meaning that $Z_a = 3$. Four nodes have $b$ as their $\chi$-image (marked with triangles), meaning that $Z_b = 4$. }
    \label{fig:Gamma_vs_gamma}
\end{figure}
Next, define a function $\chi:\bbU\to\bbU$ as follows. First, $\chi(\varnothing)=\varnothing$. Inductively, given that $\chi(v)=u$, set $\chi(v*1)=u*1$ and, for $i \ge 2$, set $\chi(v*i)=u*j$ if and only if $\tau_{v,j-1} +1< i \leq \tau_{v,j}+1$. Equivalently, for all $v \in \bbU$, $\chi(v)$ is the unique node $u$ for which $(\gamma(w))_{\varnothing \preceq w \preceq u}=(\Gamma_w)_{\varnothing \preceq w \preceq v}$ (see Figure~\ref{fig:Gamma_vs_gamma} for an example). It follows that 
$\chi(v) \in E_x(\gamma)$ if and only if $v \in E_x(\Gamma)$, which means that
\[E_x(\Gamma) = \{v \in \bbU: \chi(v) \in E_x(\gamma)\}.\]
So, writing $Z_u = |\chi^{-1}(u)|$ (again, see Figure~\ref{fig:Gamma_vs_gamma} for an example), then 
\begin{equation}
\label{P_sum-Z}
N_x(P)\leq N_x(Q)\le N_x(\Gamma)= |E_x(\Gamma)|=\sum_{u \in E_x(\gamma)}Z_u.
\end{equation}
Additionally, by the definition of $\chi$, for all $v \in \bbU$ we have $Z_{v*1}=Z_v$, and for all $i \ge 2$, 
\[
Z_{v*i}=\sum_{u \in \chi^{-1}(v)} (\tau_{v,i}-\tau_{v,i-1})\, .
\]
Using the independence of the increments $\tau_{v,i}-\tau_{v,i-1}$, it follows by induction that for all $u \in \bbU$, $Z_u$ is distributed as the size $S_{\rr(u)}$ of the $\rr(u)$'th generation in a branching process with Geometric$(1/2)$ offspring distribution, and thus (see Harris~\cite[Chapter I, Section 7.1]{Harris}) that $Z_u$ is itself Geometric$(2^{-\rr(u)})$-distributed.

\smallskip
The above fact, together with the bound $\log(1-s)\leq -s$, yields that 
\[\prob(Z_u\geq z)\leq (1-2^{-\rr(u)})^{z-1}\leq 2\exp(-2^{-\rr(u)}z)\]
for any $z\geq 1$. Lemma \ref{Bound r(u)} provides that $\sup_{u\in E_x(\gamma)}\rr(u)\leq c_1\sqrt{\log x}$, where $c_1>0$ is some constant; thus, for all $u\in E_x(\gamma)$ and $z\geq 1$,
\begin{equation}
\label{Bound Z_u}
\prob(Z_u\geq z 2^{c_1\sqrt{\log x}})
\le \prob(Z_u \ge z\cdot 2^{\rr(u)}) 
\leq 2 e^{-z}.
\end{equation}
Recalling that $N_x(\gamma)=|E_x(\gamma)|$, and using (\ref{Bound Z_u}), (\ref{P_sum-Z}) and a union bound, it follows that 
\[\prob\big(N_x(P)\geq  N_x(\gamma) \cdot z \exp(c_1\log(2) \sqrt{\log x})\big) 
    \le 2N_x(\gamma)e^{-z}.\]

Finally, choose $z=y+\log N_x(\gamma)\leq y N_x(\gamma)$. By Proposition~\ref{deter_geo_flow} 
applied with $\alpha=4/3$, we have $N_x(\gamma) \le \exp(c_2+c_2\sqrt{\log_{4/3}(x)})$, where $c_2>0$ is another constant. This yields that
\begin{align*}
& \prob\bigg(N_x(P) \geq y\exp\Big(2c_2+2c_2\sqrt{\log_{4/3}x}+c_1\log(2)\sqrt{\log x}\Big)\bigg)\\
& \le \prob\big(N_x(P)\geq N_x(\gamma) \cdot z \exp(c_1\log(2) \sqrt{\log x})\big) \\
&    \le 2N_x(\gamma)e^{-z}\\
& = 2 e^{-y}\, ,
\end{align*}
which implies the desired inequality.
\end{proof}

\subsection{Competitive ratio}

The goal of this section is to prove the following analogue of Proposition~\ref{Distribution of Phi(T)} for the $\mathrm{UA}$ model.
\begin{proposition}
\label{Distribution of Phi(T) in UA}
 For any $n\geq 1$, let $T_n \sim{\mathrm{UA}}(n)$. There exist universal constants $C, c > 0$ such that $\limsup_{n\to\infty}\prob(\Phi(T_n) \geq x) \leq Cx^{-c}$ for all $x \geq 1$. 
\end{proposition}
Before proceeding with the proof, recall from Section~\ref{presentationAndAsymptoticsUA} that $P_u \aseq \lim_{n \to \infty} |\theta_u T_n|/n$ and that there exists a family $(U_u)_{u \in \bbU}$ of independent Uniform$[0,1]$-distributed random variables such that for all $u \in \bbU$ and $i \in \N_1$, $P_{u*i}\aseq P_u\cdot U_{u*i} \prod_{j=1}^{i-1}(1-U_{u*j})$.

\begin{proof}
     For $u \in \bbU$, let $j(u) = \arg\max(U_{u*i}\prod_{j=1}^{i-1}(1-U_{u*j}), i \ge 1)$, so that $P_{u*j(u)}=\max_{i \ge 1} P_{u*i} = P_u\cdot U_{u*j(u)} \prod_{j=1}^{j(u)-1}(1-U_{u*j})$. The preceding maximum is almost surely achieved for a unique $i \ge 1$, so $j(u)$ is a.s.\ well-defined. Write $V_u=\max(U_{u*1},1-U_{u*1})\sim \mathrm{Uniform}[1/2,1]$ for $u \in \bbU$. Note that for all $u \in \bbU$ we have $U_{u*j(u)}\prod_{j=1}^{j(u)-1} (1-U_{u*j}) \le \max(U_{u*1},1-U_{u*1})=V_{u}$. Thus, for all $u \in \bbU$, we have $P_u \le \prod_{v \prec u} V_v$.

     \smallskip
     Now define a random path $(u_i)_{i \ge 0}$ in $\bbU$ as follows. Let $u_0=\varnothing$ and, inductively, for $i \ge 1$, given $u_i$ set $u_{i+1}=u_i*j(u_i)$. Since the families $((U_{u*i})_{i \ge 1}),u \in \bbU)$ are \textsc{iid}, the pairs $(V_{u};j(u))_{u \in \bbU}$ are \textsc{iid}, which in turn implies that the random variables $(V_{u_i})_{i\geq 0}$ are \textsc{iid}.
     Then let 
     \[
     K=\inf\{k \ge 1: P_{u_k} \le 1/2\} \le \inf\bigg\{g \ge 1: \prod_{i=1}^{g} V_{u_{j-1}} \le 1/2\bigg\}=G\, ,
     \]
     the final equality constituting the definition of $G$; it is straightforward to see that $G<\infty$ almost surely. Also, by (\ref{formula_P1}) each $P_u$ is a product of a finite number of independent uniform random variables so has a continuous distribution; thus, almost surely none of the $P_u$ equal $1/2$, so in particular  $P_{u_K} < 1/2$ almost surely.
     \smallskip
    
     Since $|\theta_{u} T_n|/n \convas P_{u}$ for all $u \in \bbU$ and $K$ is a.s.\ finite, it follows that there exists a random variable $n_0 \in \N_1$ such that almost surely, for all $n \ge n_0$, $|\theta_{u_i} T_n| > n/2$ for all $i \in [K-1]$. We claim that by choosing $n_0$ large enough, we may also ensure that $|\theta_{u_{K-1}*j} T_n|<n/2$ for all $j \ge 1$. (This would be immediate from the convergence fact that $|\theta_{u} T_n|/n \convas P_{u}$ for all $u$, if we were only to consider finitely many values of $j$.) To see this, let $J$ be the smallest integer $j \ge j(u_{K-1})$ for which $P_{u_{K-1}*1}+\ldots+P_{u_{K-1}*j} > P_{u_{K-1}}/2$. Then $J$ is a.s.\ finite, so for $n$ sufficiently large, 
     $|\theta_{u_{K-1}*\ell}T_n| < n/2$ for $1 \le \ell \le J$. 
     For such $n$ we also have 
     \[\frac{|\theta_{u_{K-1}*1}T_n|}{n}+\ldots+\frac{|\theta_{u_{K-1}*J}T_n|}{n}\convas P_{u_{K-1}*1}+\ldots+P_{u_{K-1}*J}> \frac{P_{u_{K-1}}}2\, ,
     \]
     so for $n$ sufficiently large
     \[
     \frac1n \sum_{\ell > J}|\theta_{u_{K-1}*\ell}T_n| 
     = 
    \frac1n\Big(|\theta_{u_{K-1}}T_n| -1-|\theta_{u_{K-1}*1}T_n|-\ldots-|\theta_{u_{K-1}*J}T_n|\Big) < \frac{P_{u_{K-1}}}{2} \le \frac12\, .
     \]
     Thus, for such $n$ we  also have $|\theta_{u_{K-1}*\ell}T_n| < n/2$ for $\ell > J$, as claimed.

     \smallskip
     It follows from the preceding paragraph that for $n \ge n_0$, $\varnothing = u_0 \preceq \dots \preceq u_{K-1}$ is the sequence of nodes described in Lemma~\ref{Upper bound Phi}, so almost surely, for $n \ge n_0$,  
     \[
     \Phi(T_n) \le \prod_{i = 1}^{K-1} \frac{1}{1 - |\theta_{u_i}T_n|/n}\,.
     \]
    Since $|\theta_{u_i}T_n|/n \convas P_{u_i}$, it follows that a.s.\ $\limsup_{n \to \infty} \Phi(T_n) \le \prod_{i=1}^{K-1} (1-P_{u_i})^{-1}$, so  
    \[
    \limsup_{n \to \infty} \prob(\Phi(T_n) \ge x) \le \prob\left(\prod_{i=1}^{K-1} \frac{1}{1-P_{u_i}} \ge x \right)
    \le 
    \prob\left(\prod_{i=1}^{G-1} \frac{1}{1-V_{u_0}\cdot \ldots \cdot V_{u_{i-1}}} \ge x \right)\, ,  
    \]
    where for the last inequality we have used that $K \le G$ and that $P_{u_i} \le \prod_{j=1}^{i} V_{u_{j-1}}$. 

    Finally, we have 
    \[\E\big[(1-V_1)^{-1/2}\big]=\int_0^1 \frac{\dd x}{\sqrt{1-\max(x,1-x)}}=2\int_0^{1/2}\frac{\dd x}{\sqrt{x}}=2\sqrt{2}\, ,\]
    so the result follows by Lemma~\ref{Bound x}.
\end{proof}

\subsection{Proof of the main theorem}

\begin{proof}[Proof of Theorem~\ref{thm:mainua}]

Recall that $|T_n|=n$: this will be used below to write some expressions more succinctly. Exactly like in the proof of Theorem~\ref{thm:dary_weak}, we partition the set $B_n=\{v\in T_n: \varphi_{T_n}(v)\leq \varphi_{T_n}(\varnothing)$\} of all better candidates than the root into subtrees stemming from children of elements of $H_n=\{v\in T_n:|\theta_v T_n|\geq \frac{1}{3}n\}$. Namely, if $v\in B_n\setminus H_n$ then it has a unique ancestor $u\in T_n\setminus H_n$ whose parent is in $H_n$. 

Here, in contrast to the proof of Theorem~\ref{thm:dary_weak}, the nodes in $H_n$ can have an arbitrarily large number of children. To control the number of children that matter, i.e.~those that have a descendant in $B_n$, we rely on Lemma~\ref{Condition on size of subtrees}. Indeed, for $u$ an ancestor of $v\in B_n$, Lemma~\ref{Condition on size of subtrees} gives us that $|\theta_v T_n|/n\geq (1+\Phi(T_n))^{-1}$ so also $|\theta_u T_n|/n\geq (1+\Phi(T_n))^{-1}$. Therefore, setting 
\[M_n=\max_{u'\in H_n} \big|\{u\in T_n : \ola{u}=u',(1+\Phi(T_n))|\theta_u T_n|\geq n\}\big|,\]
we obtain that
\begin{equation}
\label{alt_deterministic bound B_n}
|B_n|\leq |H_n|+ M_n\cdot |H_n| \cdot \max_{u\in T_n\setminus H_n}\I{(1+\Phi(T_n))|\theta_u T_n|\geq n}\big|\{v\in\theta_u T_n: u*v\in B_n\}\big|.
\end{equation} 

Like for Theorem~\ref{thm:dary_weak}, we now want to control the sizes of the sets $\{v\in \theta_u T_n: u*v\in B_n\}$, for $u$ appearing in the above maximum, by using Proposition~\ref{main}. However, unlike the bound from Lemma~\ref{flow binary}, the bound provided by Proposition~\ref{main} only holds with high probability and only concerns the asymptotic proportions of nodes lying in subtrees of $T_n$. Hence, to avoid summing the probabilities of too many bad events or trying to prove rate of convergence bounds, we will basically show that $B_n$ is contained with high probability in a deterministic finite set of words. To do so, we rely on Lemmas~\ref{Condition on size of subtrees} and~\ref{Size subtree}, and Proposition~\ref{Distribution of Phi(T) in UA}, to obtain that with high probability, all nodes $u$ with $\vp_{T_n}(u)\le\vp_{T_n}(\varnothing)$ will satisfy that $|\theta_u T_n|$ is large and $\rw(u)$ is small. This will also be useful for bounding $M_n$ and $|H_n|$. We now proceed to the details.

\smallskip
By Proposition~\ref{Distribution of Phi(T) in UA}, there exist constants $C,c_1>0$ such that for all $\varepsilon\in(0,1)$,
\begin{align}\label{Bound Phi in thm 2}
    \limsup_{n\to\infty}\prob(\Phi(T_n) \geq \eps^{-c_1})\leq C\eps.
\end{align}
Moreover, letting $c_2=\tfrac{1+2c_1}{\log(3/2)}$, Lemma~\ref{Size subtree} yields that 
\begin{equation}
    \label{finite}
    \limsup_{n\to\infty}\prob\big(\exists u\in\mathbb{U}\, :\, \rw(u)\geq c_2\log \tfrac{1}{\varepsilon}, |\theta_u T_n|\geq \tfrac{1}{2}\varepsilon^{c_1} n\big)\leq 6\varepsilon
\end{equation}
for any $\varepsilon\in(0,1)$. 

\smallskip
The above bound motivates us to define $\bbU^{(\varepsilon)}=\{v\in \bbU: \rw(v)\leq c_2\log (\sfrac1{\varepsilon})\}$; note that this set is deterministic and does not depend on $n$. Letting $E(n,\eps)$ be the event that
$\Phi(T_n)\leq \varepsilon^{-c_1}$ and that $\{u\in T_n: |\theta_u T_n|\geq \tfrac{1}{2}\varepsilon^{c_1} n\}\subset \bbU^{(\varepsilon)}$, 
we have shown that $\liminf_{n \to \infty} \prob(E(n,\eps)) \ge 1-(C+6)\eps$. We now prove several bounds assuming that $E(n,\eps)$ occurs; in the below list we work on the event $E(n,\eps)$.
\begin{enumerate}
\item Since $(1+\Phi(T_n))^{-1}\geq (2\Phi(T_n))^{-1}\geq \tfrac12\varepsilon^{c_1}$, 
if $u\in T_n$ has  $|\theta_u T_n| \ge (1+\Phi(T_n))^{-1}n$ then $u \in \bbU^{(\varepsilon)}$.
\item For $u'\in T_n$, the number of children $u$ of $u'$ with $|\theta_u T_n|\geq (1+\Phi(T_n))^{-1}n$ is bounded by the maximum of their last letters, which is itself bounded by the maximum of their weights. Hence, $M_n$ is at most the maximum weight of a node $u\in\bbU^{(\varepsilon)}$, so $M_n\leq c_2\log(\sfrac1\eps)$.
\item Since $H_n$ has at most $3$ leaves, we have $|H_n|\leq 3\max_{u\in H_n}\hgt(u)\leq 3\max_{u\in H_n}\rw(u)$, and so if $\tfrac{1}{3}\geq \tfrac{1}{2}\varepsilon^{c_1}$ then  $H_n\subset \bbU^{(\varepsilon)}$ and $|H_n|\leq  3 c_2\log(\sfrac1\eps)$.
\item For any $u \in \bbU$, by Lemma~\ref{Condition on size of subtrees}, if $u \in B_n$ then $|\theta_u T_n| \ge (1+\Phi(T_n))^{-1}n \ge \tfrac{1}{2}\varepsilon^{c_1}n$, 
so $u\in\bbU^{(\varepsilon)}$. 
It follows that if $u*v \in B_n$ then $\rw(v)\leq\rw(u*v)\leq c_2\log(\sfrac1\varepsilon)$. 
\item Finally, it holds that $\Phi(T_n)^{-1}\leq \varphi_{T_n}(u)/\varphi_{T_n}(\varnothing)$ by definition, so if $u*v\in B_n$ then  $\varphi_{T_n}(u)/\varphi_{T_n}(u*v)\geq \varepsilon^{c_1}$.
\end{enumerate}
For $\eps \in (0,1/3)$, on $E(n,\eps)$, the above bounds and (\ref{alt_deterministic bound B_n}) entail that
\begin{align}\label{deterministic bound Bn}
    |B_n|\leq 3c_2\log(\tfrac{1}{\varepsilon})+3c_2^2\log^2(\tfrac{1}{\varepsilon})\max_{u\in \bbU^{(\varepsilon)}, u\in T_n\setminus H_n }
    \bigg|\bigg\{v\in \bbU^{(\eps)}: \frac{\varphi_{T_n}(u)}{\varphi_{T_n}(u*v)}\geq \varepsilon^{c_1}\bigg\}\bigg|\, .
\end{align}

We are now almost ready to apply Proposition~\ref{main}. Let $u\in\bbU$. For all $v\in \bbU$, we define $P_v^{(u)}=P_{u*v}/P_u\aseq\lim_{n \to \infty}|\theta_{u*v}T_n|/|\theta_u T_n|$. It then follows from (\ref{eq:phiu-ancestor}) that if $P_u\leq \frac{1}{3}$ then almost surely 
\begin{align*}
\lim_{n\to\infty}\frac{\varphi_{T_n}(u)}{\varphi_{T_n}({u*v})} = \prod_{\varnothing\prec w\preceq v}\frac{P_{u*w}}{1-P_{u*w}} \leq\prod_{ \varnothing\prec w \preceq v}\frac{P^{(u)}_w}{2}
\end{align*}
for all $v \in \bbU$, and so in particular for all $v\in\bbU^{(\varepsilon)}$.
Also recall that $|\theta_u T_n|/n$ almost surely converges to $P_u$, which yields that $\limsup_{n\to\infty}\I{u\in T_n\setminus H_n}\leq \I{P_u\leq 1/3}$ almost surely. Since $\bbU^{(\varepsilon)}$ is deterministic and finite, it follows from the above bounds that almost surely
\begin{align*}
\limsup_{n\to\infty} 
\I{u\in T_n\setminus H_n}
\bigg|\bigg\{v\in \bbU^{(\eps)}: \frac{\varphi_{T_n}(u)}{\varphi_{T_n}(u*v)}\geq \varepsilon^{c_1}\bigg\}\bigg|
& \leq \bigg|\bigg\{ v \in \bbU^{(\varepsilon)}:  \prod_{\varnothing\prec w\preceq v} \frac{P^{(u)}_w}{2} \geq \eps^{c_1}\bigg\}\bigg|\\
& \le
 \bigg|\bigg\{ v \in \bbU:  \prod_{\varnothing\prec w\preceq v} \frac{P^{(u)}_w}{2} \geq \eps^{c_1}\bigg\}\bigg| \\
 & \ed N_{\eps^{-c_1}}(P)\, ,
\end{align*}
where the final identity in distribution follows from (\ref{formula_P1}) and Proposition~\ref{Independence U in UA}. To obtain a tail bound on $N_{\eps^{-c_1}}(P)$, 
we apply Proposition~\ref{main} with $y=(\log(2)c_2+1)\log(\sfrac{1}{\varepsilon})$; this yields 
\[
\prob\Big(N_{\eps^{-c_1}}(P)
\ge y\exp\big(c+c\sqrt{\log(\eps^{-c_1})}\big)\Big)
\le 2\exp(-y) = 2\eps^{c_2\log 2 + 1}. 
\]
Since $s \le \exp(\sqrt{s})$ for all $s > 0$, there exists $c_3 > 0$ such that 
$y\exp\big(c+c\sqrt{\log(\eps^{-c_1})}\big) \le \exp(c_3+c_3\sqrt{\log\tfrac1\eps})$, which together with the two preceding displays yields that
\[
\limsup_{n\to\infty}\prob\left(u\in T_n\!\setminus \! H_n\,;\left|\left\{v\in\bbU^{(\eps)} : \frac{\varphi_{T_n}(u)}{\varphi_{T_n}(u*v)}\geq \varepsilon^{c_1}\right\}\right|\geq e^{c_3 + c_3\sqrt{\log \tfrac{1}{\eps}}}\right) \leq 2^{-c_2\log 1/\eps} 2\varepsilon. 
\]
Since $|\bbU_{\eps}| \le 2^{c_2\log1/\eps}$, a union bound then gives 
\begin{equation}
\label{mass_better_outside_H}
\limsup_{n\to\infty}\prob\left(\max_{u\in \bbU^{(\eps)}} \I{u\in T_n\setminus H_n}\left|\left\{v\in\bbU^{(\eps)} : \frac{\varphi_{T_n}(u)}{\varphi_{T_n}(u*v)}\geq \varepsilon^{c_1}\right\}\right| \geq e^{c_3 + c_3\sqrt{\log \tfrac{1}{\eps}}}\right)\leq  2\varepsilon.
\end{equation}

Finally, we combine (\ref{deterministic bound Bn}) and (\ref{mass_better_outside_H}) with the fact that $\limsup_{n \to \infty} \prob(E(n,\eps)^c) \le (C+6)\eps$ to deduce that for all $\varepsilon\in(0,1/3)$,
\begin{align*}
    \limsup_{n\to\infty}\prob\Big(|B_n|\geq 3c_2\log(\tfrac{1}{\eps})+3c_2^2\log(\tfrac{1}{\eps})^2\exp\big(c_3+c_3\sqrt{\log \tfrac{1}{\eps}}\big)\,\Big)\leq  (8+C)\varepsilon.
\end{align*}
Since we always find the root when $K\geq |B_n|$, the desired result follows.
\end{proof}

\subsection{Proof of Theorem~\ref{thm:dary}}
\label{sec:dependence_d}

Let $d\geq 2$. Here, we use the notation of Section~\ref{presentationAndAsymptotics}, so that $T_n$ stands for a random plane tree with $T_n\sim\mathrm{UA}_{d+1}(n)$. In particular, recall the random variables $(D_u)_{u\in\bbU\setminus\{\varnothing\}}$ and that for $u=(u_1,\ldots,u_h)$,
\[P_u=\prod_{i=1}^h D_{(u_1,\ldots,u_i)}.\]
As already explained in Section~\ref{rem:dependence_d}, the only step of the proof of Theorem~\ref{thm:dary_weak} that introduced a dependence on $d$ in the constant $c_d^*$ is the deterministic control of the number of competitors in small subtrees given by Lemma~\ref{flow binary} and Proposition~\ref{deterministicBound}.  To remove this dependence, we can instead rely on a probabilistic analysis of the $\mathrm{UA}_{d+1}$ model and follow the strategy used for studying the general uniform attachment model more closely. Indeed, one can check that by rerunning the proof of Theorem~\ref{thm:mainua}, but using Proposition~\ref{Independence U binary} in place of Proposition~\ref{Independence U in UA} and replacing $\bbU^{(\varepsilon)}=\{u\in\bbU\, :\, \rw(u)\leq c_2\log 1/\varepsilon\}$ with the finite set $\{u\in\bbU_d\, :\, \hgt(u)\leq c_2\log 1/\varepsilon\}$, we may remove the $d$-dependence of the constants in Theorem~\ref{thm:dary_weak}, and thereby prove Theorem~\ref{thm:dary}, provided we establish the following analogue of Proposition~\ref{main}.

\begin{proposition}
\label{adaptation_main}
Let us set $P'=(P_{1*u}/P_1)_{u\in\bbU}$. Then there exists a universal constant $c>0$, that does not depend on $d$, such that for any $x,y\geq 1$,
\[\prob\big(N_x(P')\geq y\exp(c+c\sqrt{\log x})\big)\leq 2e^{-y}.\]
\end{proposition}
The reason we need $P'$ rather than $P$ in this proposition is that in the $(d+1)$-regular model, the split at the root is different from that of all other nodes in $\mathbb{U}_d$; we could equally have defined $P'=(P_{v*u}/P_v)_{u \in \bbU}$ for any $v \in \bbU_d\setminus \{\varnothing\}$.

\smallskip
Next, carefully revisiting the proofs of Proposition~\ref{main}, Lemma~\ref{Bound r(u)}, and Corollary~\ref{corr}, one can verify that Proposition~\ref{adaptation_main} follows from the adaptation of Lemma~\ref{Unif} stated below (in which, informally, the change from the $\mathrm{UA}$ to the $\mathrm{UA}_{d+1}$ model forces us to replace the values $4/3$ and $1/2=\prob(V_1\leq 3/4)$ with other universal constants $\alpha>1$ and $p=\prob(W_1\leq \alpha^{-1})>0$). 

\begin{lemma}
\label{geometric_rearrangement}
There exists a probability measure $\mu$ on $[0,1]$ with $\mu(\{1\})<1$ such that the following holds for all $d\geq 2$. If $(D_1,\ldots,D_{d})\sim \mathrm{Dir}_{d}(\tfrac{1}{d-1})$, then there is a sequence $(W_i)_{i\geq 1}$ of \textsc{iid}~random variables with distribution $\mu$ and a random permutation $\sigma$ of $[d]$ such that almost surely $D_{\sigma(i)}\leq \tfrac{1}{2}\prod_{j=1}^{i-2} W_j$ for all $2\leq i\leq d$.
\end{lemma}

We now explain how to prove Lemma~\ref{geometric_rearrangement}. We seek an admissible probability measure $\mu$ of the form $p\delta_{\alpha^{-1}}+(1-p)\delta_{1}$ with $\alpha>1$ and $p>0$. Our argument is based on a sampling method for the symmetric Dirichlet distribution via \emph{stick-breaking} which builds the components of the random vector in an expression similar to $U_i\cdot\prod_{j=1}^{i-1}(1-U_j)$. This will put us in a good position to use the proof method of Lemma~\ref{Unif}.

\smallskip

Let $\pi$ be a uniformly random permutation of $[d]$. Independently from $\pi$, let $X_1,\ldots,X_{d-1}$ be independent random variables such that $X_i\sim\mathrm{Beta}(\tfrac{d}{d-1},1-\tfrac{i-1}{d-1})$ for all $i\in[d-1]$. Then, set $Y_{d}=\prod_{j=1}^{d-1}(1-X_j)$ and $Y_i=X_i\cdot \prod_{j=1}^{i-1}(1-X_j)$ for all $i\in[d-1]$. It holds that
\begin{equation}
\label{stick_breaking_exchangeable}
(Y_{\pi(1)},\ldots,Y_{\pi(d)})\sim\mathrm{Dir}_{d}\big(\tfrac{1}{d-1}\big)\, .
\end{equation}
This identity in distribution can be derived from the fact that for any integers $j_1,\ldots,j_k\in\N_1$, $\mathrm{Dir}_k(\tfrac{j_1}{d-1},\ldots,\tfrac{j_k}{d-1})$ is the limit law of the proportion-vector in a $k$-colour P\'olya urn starting with $j_i$ balls of the $i$'th colour for all $i\in[k]$, and where each ball drawn is returned along with $d-1$ new balls of the same colour; see e.g.~\cite{athreya_lost} or \cite[Section 6]{CMP15}. Then (\ref{stick_breaking_exchangeable}) describes the case where $j_1=\ldots=j_{d-1}=1$: the permutation $\pi^{-1}$ represents the order in which the colours are drawn for the first time and for $i\in[d-1]$, $X_i$ is the relative proportion of colour $\pi^{-1}(i)$ with respect to the colours that have yet to be drawn.

\smallskip
Thanks to (\ref{stick_breaking_exchangeable}), we only need to construct \textsc{iid}~random variables $W_1,\ldots, W_{d-2}$ with common distribution $\mu=p\delta_{\alpha^{-1}}+(1-p)\delta_{1}$ and a random permutation $\sigma$ of $[d]$ such that $Y_{\sigma(i)}\leq \tfrac{1}{2}\prod_{j=1}^{i-2} W_j$ for all $2\leq i\leq d$. Similarly to as in the proof of Lemma~\ref{Unif}, we consider $K=\min \{d\}\cup \{i\in[d-1]\, :\, X_i>\tfrac{1}{2}\}$ and define $\sigma(i)=K\I{i=1}+(i-1)\I{2\leq i\leq K}+i\I{i\geq K+1}$ for all $i\in[d]$. Then, to be able to obtain the same deterministic inequalities as in the proof of Lemma~\ref{Unif}, we must construct the $W_1,\ldots,W_{d-2}$ so that it holds that 
\[W_i\geq 
\begin{cases}
    1-X_i&\text{ if }i\leq K-1,\\
    1-X_{i+1}&\text{ if }i\geq K,
\end{cases}\]
almost surely for all $i\in[d-2]$ while ensuring their independence. By independence of $X_1,\ldots,X_{d-1}$, this is possible as soon as for all $i\in[d-2]$, \[\prob(1-X_i>\alpha^{-1}\, |\, X_i\leq 1/2)\leq 1-p\quad\text{ and }\quad\prob(1-X_{i+1}>\alpha^{-1})\leq 1-p.\]

It follows from the above discussion that Lemma~\ref{geometric_rearrangement} --- and, hence, Theorem~\ref{thm:dary}
--- is a consequence of the following uniform estimate for Beta distributions. 
\begin{equation}
\label{estimate_tails_beta_conditioned}
\forall a\in[1,2],\forall b\in(0,1],\quad \text{if}\quad X\sim\mathrm{Beta}(a,b)\quad\text{then}\quad \prob(X<1/17\, |\, X\leq 1/2)\leq \frac{1}{2}.
\end{equation}
More concretely, this bound yields that Lemma~\ref{geometric_rearrangement} holds with $\mu=\tfrac{1}{2}\delta_{16/17}+\tfrac{1}{2}\delta_1$.

\begin{proof}[Proof of (\ref{estimate_tails_beta_conditioned})]
Let $B,B'$ be two random variables such that $B'\sim\mathrm{Beta}(1,b)$ and $B\sim\mathrm{Beta}(2,b)$. From (\ref{eq:betadomination}), we have the stochastic dominations $B'\preceq_{\mathrm{st}} X\preceq_{\mathrm{st}} B$, and so 
\[\prob(X < 1/17\, |\, X\leq 1/2)\leq\frac{\prob(B'< 1/17)}{\prob(B\leq 1/2)}\, .\]
We bound the numerator and the denominator separately. Since $\Gamma(1+b)=b\Gamma(b)\Gamma(1)$,
\[\prob(B'< 1/17)=b\int_0^{1/17}(1-x)^{b-1}\, \dd x \leq b\int_0^{1/17}(1-x)^{-1}\, \dd x \leq \frac{b}{17}(16/17)^{-1}=\frac{b}{16}\]
because $0<b\leq 1$. Similarly, since $\Gamma(2+b)=(1+b)\Gamma(1+b)=(1+b)b\Gamma(b)\Gamma(2)$,
\[\prob(B\leq 1/2)=(1+b)b\int_0^{1/2} x (1-x)^{b-1}\, \dd x \geq b\int_0^{1/2}x \, \dd x =\frac{b}{8},\]
which concludes the proof.
\end{proof}

\begin{acks}[Acknowledgments]
LAB, CF and ST acknowledge support from NSERC during the preparation of this research. RK has also been supported by NSERC via a Banting postdoctoral fellowship [BPF-198443]. LAB additionally acknowledges support from the Canada Research Chairs program. LRL acknowledges support from the Institut des Sciences Math\'ematiques. All authors thank Th\'eodore Conrad-Frenkiel for useful discussions about the project. 
\end{acks}

\begin{appendix}

\section{Proof of the tightness of Theorem~\ref{thm:dary}}
\label{appendix_tightness}

Let $d\geq 3$ be a positive integer, and let $(T_n)_{n\geq 1}$ follow the $\mathrm{UA}_d$ model. Recall that this means that $T_1$ consists of a single root node $\oo$ and $d$ leaves, and that for each $n \ge 1$, $T_{n+1}$ is built from $T_n$ by adding $d-1$ new neighbours to a single leaf of $T_n$ chosen uniformly at random. In this appendix, we show that any root-finding algorithm for $(T_n)$ with error less than $\varepsilon>0$ must have size bigger than $C^*\exp(c^*\sqrt{\log\sfrac1\eps})$ for some constants $C^*,c^*>0$.

\begin{proposition}\label{prop_appendix}
    For any $d\geq 3$, there exist $C^*,c^*>0$ such that the following holds. For all $\eps>0$, any root-finding algorithm $A$ such that $\limsup_{n\to\infty} \prob(\oo\notin A(T_n)) \leq \varepsilon$ must have size $s(A)\geq C^*\exp(c^*\sqrt{\log\sfrac1\eps})$, where recall that $s(A)=\sup\{ |A(t)| : t\mbox{ a finite tree}\}$.
\end{proposition}

To prove this proposition, we work with the equivalent model discussed in Remark~\ref{remove leaves}, i.e., the diffusion process on the infinite $d$-regular tree. Namely, for each $n\geq 1$, let $D_n$ be (deterministically) obtained from $T_n$ by removing all the leaves of $T_n$, so that $D_n$ and $T_n$ share the same root. Conversely, $T_n$ can be (deterministically) recovered from $D_n$ by attaching, for each vertex $u$ of $D_n$, $d-\mathrm{deg}_n(u)$ leaves to $u$, where $\mathrm{deg}_n(u)$ stands for the degree of $u$ in $D_n$. Therefore, Proposition~\ref{prop_appendix} is a direct consequence of the following result.

\begin{lemma}\label{lem_appendix}
    For any $d\geq 3$, there exist $C^*,c^*>0$ such that the following holds. For all $\eps>0$, any root-finding algorithm $A$ such that $\limsup_{n\to \infty} \prob(\oo\notin A(D_n)) \leq \varepsilon$ must have size $s(A)\geq C^*\exp(c^*\sqrt{\log\sfrac1\eps})$.
\end{lemma}

Observe that the random sequence $(D_n)_{n\geq 1}$ can be sampled as follows: $D_1$ consists of a single root node $\oo$, and for each $n\geq 1$, $D_{n+1}$ is built from $D_{n}$ by adding a single new neighbour to a vertex $u$ of $D_n$ chosen with probability proportional to $\max(0,d-\mathrm{deg}_n(u))$. 
From this recursive scheme, we know from Crane and Xu~\cite[Theorem~3]{Inference} that the algorithm $\mathcal{A}_K$ (which outputs the $K$ most central vertices in the sense of (\ref{def phi in graph})) is the optimal size-$K$ root-finding algorithm for $D_n$ --- specifically, $\varphi_{D_n}$ is inversely proportional to the likelihood function of being the root~\cite[Eq.~(9)]{Rumors}. Thus, to prove Lemma~\ref{lem_appendix}, we precisely need to show that there exist two constants $C^*,c^*>0$ such that for all $\varepsilon>0$ and $K\geq 1$,
\[\limsup_{n\to\infty} \prob(\oo \notin \mathcal{A}_K(D_n))\leq \varepsilon \Longrightarrow K\geq C^*\exp(c^*\sqrt{\log\sfrac1\eps}).\]

\begin{proof}[Proof of Lemma~\ref{lem_appendix}]

    This proof is an adaptation of the proof of \cite[Theorem 4]{FindingAdam}. In what follows we omit floors and ceilings for readability. Moreover, we label the nodes of the $(D_n)_{n\geq 1}$ by their order of arrival, so that $\oo=1$ and $n+1$ is the only node of $D_{n+1}$ that do not belong to $D_n$, for each $n\geq 1$. 
    
    Set $K= \exp\left(\sqrt{\frac{1}{30}\log\frac{1}{2\eps}}\right)$, which is assumed to be bigger than some universal threshold $K_0>0$ to be determined in the course of the proof. By the optimality of $\mathcal{A}_K(D_n)$, it suffices to show that we have $\varepsilon< \limsup_{n\to\infty}\prob(1 \notin \mathcal{A}_K(D_n))$. Since $\prob(1 \notin \mathcal{A}_K(D_n))$ is non-decreasing in $n$ (again by optimality), we only need to show that $\prob(1 \notin \mathcal{A}_K(D_{K+1}))> \varepsilon$. To do so, we consider a set $\mathcal{T}$ of labeled trees with vertices $\{1,\ldots,K+1\}$, rooted at $1$, that satisfies $\mathbb{P}(D_{K+1}\in \mathcal{T}) > \varepsilon$. We will then show that for any $T\in \mathcal{T}$, $\varphi_T(1)>\varphi_T(v)$ for all nodes $v$ of $T$ distinct from $1$.
    \smallskip

    The trees in the set $\mathcal{T}$ are those labeled trees with vertices $\{1,\ldots,K+1\}$, rooted at $1$, that satisfy the following conditions:
    \begin{enumerate}
        \item[(C1)] The oldest $10\log K$ nodes form a path.
        \item[(C2)] The remaining vertices are descendants of the node with label $10\log K$.
        \item[(C3)] The height of the subtree defined in (C2) is at most $8\log K $.
    \end{enumerate}
    Each tree in $\mathcal{T}$ corresponds to a path of length $10\log K$, with the root at one end and a subtree of the infinite $d$-regular tree attached at the other end. We now show that $\prob(D_{K+1}\in \mathcal{T}) >\eps$. The probability that in the process $(D_n, n \ge 1)$, (C1) occurs as prescribed above is
    \[\prod_{i=2}^{10\log K-1}\frac{d-1}{i(d-2)+2} \geq \prod_{i=2}^{10\log K-1}\frac{d-1}{i(d-1)}=\frac{1}{(10\log K-1)!} > \exp(- 10\log^2K),\]
    the last inequality holding provided that $K$ is large enough. \smallskip
    
    Similarly, given that (C1) has occurred, the probability that $D_{K+1}$ satisfies (C2) is
    \[\prod_{i=1}^{K+1-10\log K}\frac{i(d-2)+1}{(i+10\log K-1)(d-2)+2}  >\exp(- 20\log^2K),\]
    the last inequality holding provided that $K$ is large enough. \smallskip
    
    Finally, the conditional probability that $D_{K+1}$ satisfies (C3) given that it satisfies (C1) and (C2) is larger than $\sfrac{1}{2}$ for large enough $K$. Indeed, given that (C1) and (C2) occur, the subtree defined in (C2) is distributed as {\color{black} a \emph{uniform $(d-1)$-ary search tree} of size $K+2-10\log K$,\footnote{A \emph{uniform $(d-1)$-ary search tree of size $n$} is a random rooted tree $D_n'$, which is a slight variant of $D_n$ whose root has degree at most $d-1$ instead of $d$. Namely, $(D_n')_{n\geq 1}$ can be recursively sampled as follows:~$D_1'$ consists of a single root vertex, and for each $n\geq 1$, $D_{n+1}'$ is built from $D_{n}'$ by adding a single new child to a vertex $u$ of $D_n$ chosen with probability proportional to $\max(0,d-1-\rk_u(D_n'))$, where $\rk_u(D_n')$ stands for the number of children of $u$ in $D_n'$.} and by the main theorem of \cite{devroye1990height}, we know that the height of a uniform $(d-1)$-ary search tree of size $n$ scaled by $\log(n)$ converges in probability to some $\gamma<8$ as $n \to \infty$.} Thus, 
    \[\prob(D_{K+1}\in \mathcal{T}) > \frac{1}{2}\exp(-10\log^2K)\exp(-20\log^2K)=\varepsilon,\]
    by the choice of $K$.
    \smallskip
    
    We next show that for $T\in \mathcal{T}$, we have $\varphi_T(1)>\varphi_T(i)$ for all nodes $i\in \{2,\ldots,K+1\}$. For convenience, and analogously as in Section~\ref{words_section}, for $i,j\in\{1,\ldots,K+1\}$, we write $j\preceq i$ when $i$ is a descendant of $j$ in $(T,1)$ --- i.e., $j$ is on the path of $T$ between $1$ and $i$ --- and we define $\theta_i T$ as the set of the descendants of $i$ in $(T,1)$. We write $j\prec i$ when $j\preceq i$ and $j\neq i$. By (\ref{eq:phiu-ancestor}), it suffices to show that for all $i\in\{2,\ldots,K+1\}$,
    \begin{equation}
    \label{eq:last_step_appendix}
    \prod_{1\prec j\preceq i}|\theta_j T|> \prod_{1\prec j\preceq i}(|T|-|\theta_j T|).
    \end{equation}
    For any node $j$ such that $1\prec j\preceq 10\log K$, we have that $|\theta_j T|\geq |T|-10\log K$, due to the subtree formed in (C2). Provided that $K$ is large enough so that $20\log K<K+1=|T|$, it follows that $|\theta_j T|>|T|-|\theta_j T|$. Thus, (\ref{eq:last_step_appendix}) holds for all nodes $i$ belonging to the path of length $10\log K$ formed in (C1). \smallskip  
    
    For the nodes $i$ in the subtree $\theta_{10\log K}T$ rooted at the end of the path of length $10\log K$, the set $\{j:1\prec j\prec i\}$ is partitioned by $\{j:1\prec 10\log K\}$ and $\{10\log K\prec j\preceq i\}$. By (C1) and (C2), the product $\prod_{1\prec j\preceq 10\log K}(|T|-|\theta_j T|)$ is equal to $(10\log K-1)!$. For the second product, (C3) yields that there are at most $8\log K$ vertices $j$ such that $10\log K\prec j\preceq i$. Since we always have that $|T|-|\theta_jT|\leq K+1$, we can write 
    \[\prod_{1\prec j\preceq i}(|T|-|\theta_j T|)\leq (K+1)^{8\log K}(10\log K)!\leq (K+1)^{8\log K}(10\log K)^{10\log K}.\]
    Furthermore, we also have that
    \[\prod_{1\prec j\preceq i}|\theta_j T|\geq \prod_{1\prec j\preceq 10\log K}|\theta_j T|\geq (K+1-10\log K)^{10\log K},\]
    which is strictly bigger than $(K+1)^{8\log K}(10\log K)^{10\log K}$ for $K$ large enough. This shows (\ref{eq:last_step_appendix}), and thus completes the proof.   
\end{proof}
\end{appendix}

\bibliographystyle{imsart-number}
\bibliography{biblio}

@article{beta_order_modern,
author = {Arab, Idir and Oliveira, Paulo Eduardo and Wiklund, Tilo},
title = {{Convex transform order of Beta distributions with some consequences}},
journal = {Statistica Neerlandica},
volume = {75},
number = {3},
pages = {238 -- 256},
url = {https://onlinelibrary.wiley.com/doi/abs/10.1111/stan.12233},
eprint = {https://onlinelibrary.wiley.com/doi/pdf/10.1111/stan.12233},
year = {2021}
}

@article{devroye1990height,
  author  = {Devroye, Luc},
  title   = {On the height of random m-ary search trees},
  journal = {Random Structures \& Algorithms},
  volume  = {1},
  number  = {2},
  pages   = {191--204},
  year    = {1990},
}

@article{beta_order_old,
author = {Bernd Lisek},
title = {Comparability of special distributions},
journal = {Series Statistics},
volume = {9},
number = {4},
pages = {587 -- 598},
year = {1978},
publisher = {Taylor \& Francis},
doi = {10.1080/02331887808801456},
URL = { https://doi.org/10.1080/02331887808801456},
eprint = {https://doi.org/10.1080/02331887808801456}
}

@article{AnalysisOfCentrality,
author = {Sayan Banerjee and Shankar Bhamidi},
title = {{Root finding algorithms and persistence of Jordan centrality in growing random trees}},
volume = {32},
journal = {The Annals of Applied Probability},
number = {3},
publisher = {Institute of Mathematical Statistics},
pages = {2180 -- 2210},
year = {2022},
URL = {https://doi.org/10.1214/21-AAP1731}
}

@article{Persistence,
author={Sayan,Banerjee and Shankar,Bhamidi},
year={2021},
month={08},
title={Persistence of hubs in growing random networks},
journal={Probability Theory and Related Fields},
volume={180},
number={3-4},
pages={891 -- 953},
isbn={01788051},
url={https://link.springer.com/article/10.1007/s00440-021-01066-0},
}

@article{Degree,
author = {Sayan Banerjee and Xiangying Huang},
title = {{Degree centrality and root finding in growing random networks}},
volume = {28},
journal = {Electronic Journal of Probability},
publisher = {Institute of Mathematical Statistics and Bernoulli Society},
pages = {1 -- 39},
year = {2023},
URL = {https://doi.org/10.1214/23-EJP930}
}

@article{Archaeology,
title={{Archaeology of random recursive dags and Cooper-Frieze random networks}},
volume={32},
number={6},
journal={Combinatorics, Probability and Computing},
author={Briend, Simon and Calvillo, Francisco and Lugosi, Gábor},
year={2023},
pages={859 -– 873}
}

@article{FindingAdam,
author = {Bubeck, Sébastien and Devroye, Luc and Lugosi, Gábor},
title = {Finding {A}dam in random growing trees},
journal = {Random Structures \& Algorithms},
volume = {50},
number = {2},
pages = {158 -- 172},
url = {https://onlinelibrary.wiley.com/doi/abs/10.1002/rsa.20649},
eprint = {https://onlinelibrary.wiley.com/doi/pdf/10.1002/rsa.20649},
year = {2017}
}

@article{Influence,
author={Bubeck, Sébastien and Mossel, Elchanan and Rácz, Miklós Z.},
journal={IEEE Transactions on Network Science and Engineering}, 
title={On the Influence of the Seed Graph in the Preferential Attachment Model}, 
year={2015},
volume={2},
number={1},
pages={30 -- 39},
}

@article{Scaling,
author = {Nicolas Curien and Thomas Duquesne and Igor Kortchemski and Ioan Manolescu},
title = {Scaling limits and influence of the seed graph in preferential attachment trees},
journal = {Journal de l{\textquoteright}\'Ecole polytechnique {\textemdash} Math\'ematiques},
pages = {1 -- 34},
publisher = {\'Ecole polytechnique},
volume = {2},
year = {2015},
mrnumber = {3326003},
zbl = {1320.05110},
url = {https://jep.centre-mersenne.org/articles/10.5802/jep.15/}
}

@article{Eve,
author = {Contat, Alice and Curien, Nicolas and Lacroix, Perrine and Lasalle, Etienne and Rivoirard, Vincent},
year = {2024},
month = {01},
pages = {321 -- 336},
title = {Eve, {A}dam and the preferential attachment tree},
volume = {190},
journal = {Probability Theory and Related Fields},
}

@article{Inference,
author = {Crane, Harry and Xu, Min},
title = {Inference on the History of a Randomly Growing Tree},
journal = {Journal of the Royal Statistical Society Series B: Statistical Methodology},
volume = {83},
number = {4},
pages = {639 -- 668},
year = {2021},
month = {07},
issn = {1369-7412},
url = {https://doi.org/10.1111/rssb.12428},
eprint = {https://academic.oup.com/jrsssb/article-pdf/83/4/639/49320709/jrsssb\_83\_4\_639.pdf}
}

@article{discovery,
author = {Reddad, Tommy and Devroye, Luc},
journal = {Internet Mathematics},
number = {1},
year = {2019},
month = {02},
title = {On the discovery of the seed in uniform attachment trees},
volume = {1}
}

@article{Erdos,
 ISSN = {0003486X, 19398980},
 URL = {http://www.jstor.org/stable/1968802},
 author = {Paul Erd\H{o}s},
 journal = {Annals of Mathematics},
 number = {3},
 pages = {437 -- 450},
 publisher = {[Annals of Mathematics, Trustees of Princeton University on Behalf of the Annals of Mathematics, Mathematics Department, Princeton University]},
 title = {On an Elementary Proof of Some Asymptotic Formulas in the Theory of Partitions},
 urldate = {2024-11-20},
 volume = {43},
 year = {1942}
}

@article{looking,
author = {Alan Frieze and Wesley Pegden},
title = {{Looking for vertex number one}},
volume = {27},
journal = {The Annals of Applied Probability},
number = {1},
publisher = {Institute of Mathematical Statistics},
pages = {582 -- 630},
year = {2017},
URL = {https://doi.org/10.1214/16-AAP1212}
}

@article{Confidence,
author={Khim, Justin and Loh, Po-Ling},
journal={IEEE Transactions on Network Science and Engineering}, 
title={Confidence Sets for the Source of a Diffusion in Regular Trees}, 
year={2017},
volume={4},
number={1},
pages={27 -- 40},
}

@article{Rumors,
author={Shah, Devavrat and Zaman, Tauhid},
journal={IEEE Transactions on Information Theory}, 
title={{Rumors in a Network: Who's the Culprit?}}, 
year={2011},
volume={57},
number={8},
pages={5163 -- 5181},
}

@article{athreya_lost,
author = {Athreya, Krishna B.},
title = {{On a characteristic property of P\'olya’s urn}},
journal = {Studia Scientiarum Mathematicarum Hungarica},
volume = {4},
pages = {31 -– 35},
year = {1969},
MRNUMBER = {0247643}
}

@Article{CMP15,
author={Brigitte Chauvin and Cécile Mailler and Nicolas Pouyanne},
title={{Smoothing Equations for Large Pólya Urns}},
journal={Journal of Theoretical Probability},
year={2015},
volume={28},
number={3},
pages={923 -- 957},
month={09},
url={https://ideas.repec.org/a/spr/jotpro/v28y2015i3d10.1007_s10959-013-0530-z.html}
}

@book{Harris,
author={Harris, Theodore E.},
title={{The Theory of Branching Processes}},
address= {Santa Monica, CA},
year= {1964},
publisher= {RAND Corporation}
}

\end{document}